\documentclass{article}

\usepackage{graphicx}
\usepackage{amsmath}
\usepackage{amsthm}
\usepackage{amssymb}
\usepackage{amsfonts}

\usepackage{stmaryrd}
\usepackage{proof}
\usepackage{mathtools}
\usepackage{stmaryrd}
\usepackage{upgreek}
\usepackage[frozencache=true,cachedir=minted-cache]{minted} 
\usepackage{multicol}
\usepackage{booktabs}
\usepackage{enumitem}
\usepackage{mathtools}
\usepackage{stackengine}
\usepackage{todonotes}
\usepackage{tikz-cd}

\usepackage{svg}
\usepackage[caption=false]{subfig}
\usepackage{url}
\usepackage{hyperref}
\usepackage{arydshln}

\newcommand{\IN}{\mathbb{N}}
\newcommand{\IR}{\mathbb{R}}
\newcommand{\IQ}{\mathbb{Q}}
\newcommand{\bR}{\mathbf{R}}
\newcommand{\IZ}{\mathbb{Z}}
\newcommand{\IC}{\mathbb{C}}
\newcommand{\Casm}{\mathsf{Asm}(\mathcal{K}_2)}
\newcommand{\Cset}{\mathsf{Set}}
\newcommand{\Baire}{\IN^\IN}

\newcommand{\cff}{\mathit{ff}}
\newcommand{\ctt}{\mathit{tt}}
\newcommand{\abs}[1]{\left|#1\right|}

\newcommand{\terminalobj}{\mathbf{1}}
\newcommand{\sem}[1]{\llbracket #1 \rrbracket}

\newcommand{\semmval}{\mathbf{M}}
\newcommand{\semselect}{\mathbf{select}}
\newcommand{\semGamma}{\mathbf{\Gamma}}
\newcommand{\semnabla}{\boldsymbol{\nabla}}

\newcommand{\bA}{\mathbf{A}}

\newcommand{\bK}{\mathbf{K}}

\newcommand{\dR}{\mathsf{R}}
\newcommand{\dRc}{\tilde{\mathsf{R}}}
\newcommand{\dK}{\mathsf{K}}
\newcommand{\dN}{\mathsf{N}}
\newcommand{\dZ}{\mathsf{Z}}
\newcommand{\dC}{\mathsf{C}}

\newcommand{\Type}{\mathsf{Type}}
\newcommand{\ttrue}{\mathsf{True}}
\newcommand{\tfalse}{\mathsf{False}}
\newcommand{\Prop}{\mathsf{Prop}}

\newcommand{\kleene}{\textsf{K}}
\newcommand{\unit}{\textsf{unit}}
\newcommand{\mval}{\mathsf{M}}
\newcommand{\cval}{{\mathsf{P}_+}}
\newcommand{\mult}{\mathsf{mult}}
\newcommand{\lift}{\mathsf{lift}}

\newcommand{\nablaunit}{\textsf{unit}_{\nabla}\!}
\newcommand{\nablalift}[1]{\mathrel{{#1}^{\dagger_{\nabla}}}}

\newcommand{\pic}{\mathsf{picture}}
\newcommand{\picc}{\mathsf{pic}}

\newcommand{\all}[1]{\mathrm{\Pi}( #1).\ }
\newcommand{\allx}{\mathrm{\Pi}}
\newcommand{\some}[1]{\mathrm{\Sigma}( #1).\ }
\newcommand{\somex}{\mathrm{\Sigma}}
\newcommand{\csome}[1]{\exists(#1 ).\ }
\newcommand{\usome}[1]{\exists!(#1 ).\ }
\newcommand{\csomex}{\exists}
\newcommand{\lam}[1]{\lambda ( #1 ).\ }

\newcommand{\coprodp}[2]{ #1  + #2}
\newcommand{\coprodc}[2]{ #1  + #2}

\newcommand{\defined}[1]{#1 \!\downarrow}

\newcommand{\ctrue}{\textsf{true}}
\newcommand{\cfalse}{\textsf{false}}
\newcommand{\select}{\textsf{select}}

\newcommand{\countablelift}{\mathsf{\omega lift}}

\newcommand{\upc}[1]{\lceil #1 \rceil}
\newcommand{\downc}[1]{\lfloor #1 \rfloor}

\newcommand{\semidec}{\textsf{semiDec}}

\newcommand{\singletonextract}{\textsf{elimM}}
\newcommand{\mto}{\rightrightarrows}

\newcommand{\relator}{\textsf{relator}}
\newcommand{\code}[1]{\text{\texttt{#1}}}
\newcommand{\kneg}{\hat\neg}
\newcommand{\klor}{\mathrel{\hat{\lor}}}
\newcommand{\kland}{\mathrel{\hat{\land}}}

\newcommand{\relaterule}{\textsc{Relate}}
\newcounter{mycounter}

\newcommand\letvdash[1]{\mathrel{
  \stackengine{1ex}{\vdash}{\scriptscriptstyle#1}{O}{c}{F}{T}{L}}}
\stackMath

\newcommand{\classicaltyping}{\letvdash{\sim}}
\newcommand{\constructivetyping}{\letvdash{}}

\newtheorem{lemma}{Lemma}
\newtheorem{theorem}{Theorem}
\newtheorem{metatheorem}{Metatheorem}
\theoremstyle{remark} 
\newtheorem{remark}{Remark}
\newtheorem{fact}{Fact}
\theoremstyle{definition} 
\newtheorem{example}{Example}
\newtheorem{definition}{Definition}
\newtheorem{corollary}{Corollary}

\tikzset{
    labl/.style={anchor=south, rotate=90, inner sep=.5mm}
}
\tikzcdset{every label/.append style = {font = \small}}

\providecommand{\keywords}[1]
{
  \small	
  \textbf{\textit{Keywords---}} #1
}

\title{
Extracting efficient exact real number computation from proofs in constructive type theory\thanks{Holger Thies is supported by JSPS KAKENHI Grant Number JP20K19744.
 Sewon Park is supported by JSPS KAKENHI Grant number JP18H03203.
 \includegraphics[scale=0.04]{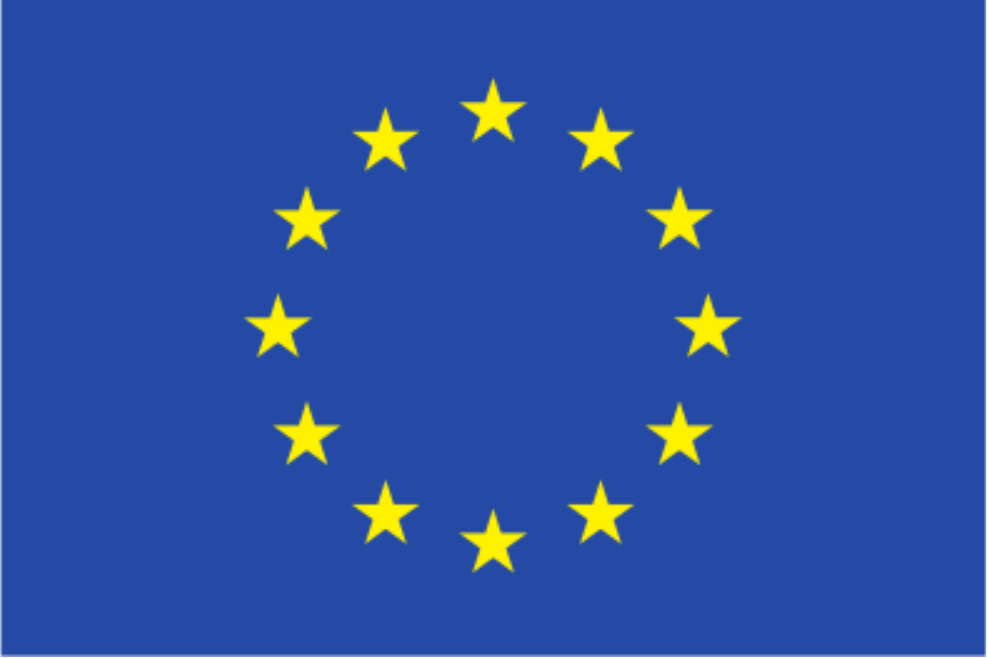} This project has received funding from the EU’s Horizon 2020 research and innovation programme under the Marie Skłodowska-Curie grant agreement No 731143.
 The authors thank Franz Brau{\ss}e and Norbert M\"{u}ller for helpful discussions.}}

\author{
Michal Kone\v{c}n\'{y}$^1$,
Sewon Park$^2$,
Holger Thies$^2$ \\
\small $^{1}$Aston University, UK\\
\small $^{2}$Kyoto University, Japan
}

\date{}
\begin{document}

\maketitle

\begin{abstract}
Exact real computation is an alternative to floating-point arithmetic where operations on real numbers are performed exactly, without the introduction of rounding errors.
When proving the correctness of an implementation, one can focus solely on the mathematical properties of the problem without thinking about the subtleties of representing real numbers.
We propose a new axiomatization of the real numbers in a dependent type theory with the goal of extracting certified exact real computation programs from constructive proofs.
Our formalization differs from similar approaches, in that we formalize the reals in a conceptually similar way as some mature implementations of exact real computation.
Primitive operations on reals can be extracted directly to the corresponding operations in such an implementation, producing more efficient programs.
We particularly focus on the formalization of partial and nondeterministic computation, which is essential in exact real computation.

We prove the soundness of our formalization with regards of the standard realizability interpretation from computable analysis and show how to relate our theory to a classical formalization of the reals.
We demonstrate the feasibility of our theory by implementing it in the Coq proof assistant and present several natural examples.
From the examples we have automatically extracted Haskell programs that use the exact real computation framework AERN for efficiently performing exact operations on real numbers. 
In experiments, the extracted programs behave similarly to  native implementations in AERN in terms of running time and memory usage.
\end{abstract}


\tableofcontents
\section{Introduction}
Exact real computation is an elegant programming paradigm in which real numbers and other continuous mathematical structures are treated as basic entities that can be manipulated exactly.
This is usually realized by adding a datatype for reals and arithmetic operations on them as primitives in programming languages.
All computations on real numbers are done exactly, without introducing rounding errors, and the implementation usually allows to output finite approximations of real numbers up to any desired absolute precision.
Although the approach is less efficient than optimized hardware-based floating-point calculations, implementations in exact real computation are by design more reliable than the former which is desirable in safety-critical application domains where correctness is more important than efficiency.
Furthermore, comparably efficient implementations exist that try to minimize the additional overhead \cite{irram,balluchi2006ariadne,konecny2008aern}. 
These implementations are successfully used to approximate numerical values up to very high precision in practice.

Proof assistants and automatic verification methods are increasingly used in computer science to verify the correctness of hardware and software \cite{ringer2020qed}.
Numerical computations, often at the intersection between software and the physical world, are in many cases safety-critical and thus their verification is a very active area of research \cite{melquiond2008proving,boldo2009,boldo2011flocq,boldo2019round,gallois2020optimal}.
Verifying floating point programs is, however, challenging as one not only has to make sure that the mathematical statements are correct, but also needs to deal with issues of the approximate implementation of real numbers  \cite{monniaux2008pitfalls,DBLP:books/daglib/0041425}.
In exact real computation, on the other hand, the implementation behaves closely to the mathematical theory, and one can therefore solely focus on the purely mathematical statements, which facilitates the process of formal verification.

When it comes to formally verifying programs, one can broadly distinguish between two main approaches  \cite{chlipala2013certified}.
The first one is classical program verification, that is to first write the program and then, in a second step, formally verify that it meets some specification. 
The second one is to extract programs directly from constructive proofs, that is, to integrate programming, specification and verification in a single step.
Arguably, the second approach is more elegant and, due to peculiarities of computing with real numbers, might even be more natural when dealing with exact real computation. 
Indeed, implementations of constructive real numbers exist in many modern proof assistants (see e.g. \cite{BLM16} for an overview), and there are numerous works solely dealing with the extraction of exact real computation programs, e.g. \cite{miyamoto2015program,DBLP:journals/apal/BergerT21}.
Still, giving up on writing programs directly, some control over the exact form of the algorithm is lost and thus fine-grained optimization becomes more difficult.
Results regarding extraction of exact real computation programs are therefore mostly theoretical and do not produce code that is efficient enough to be used in practice \cite{Mueller18}. 
Moreover, subtleties like partiality, concurrency and nondeterministic computation that arise naturally and are unavoidable in exact real computation have been less studied in the context of program extraction \cite{berger2016extracting}.

Recently, the program verification approach has therefore also gained popularity, in particular in the computable analysis community.
For example, \cite{park2016foundation} formalizes a simple imperative programming language for exact real computation as a framework for program verification and \cite{steinberg2019quantitative} formalizes some of the theory of classical computable analysis to verify functional implementations of real number computations in the Coq proof assistant.

In this work, we suggest a different approach to overcome some of the above mentioned problems of program extraction.
Instead of formalizing a concrete implementation of constructive real numbers, we suggest a high-level axiomatization of a constructive real number type, formalizing real numbers conceptually similarly to how they are treated in some of the more efficient implementations of exact real computation.
The goal is to extract programs where the constructive real type, and its axiomatically defined basic operations such as arithmetic or limits, are directly mapped to corresponding types and operations in an exact real computation framework.
That is, instead of reimplementing and verifying basic real number operations, we trust the implementation of a core of simple real number operations and to verify programs using those operations under the assumption that they are correctly implemented.
Concretely, utilising this mapping and program extraction techniques, we obtain certified programs over an implementation of exact real computation from correctness proofs.
The approach also provides a certain amount of independence of the concrete implementation of real numbers and thus allows to easily switch the underlying framework.

Some preliminary results of this work have been presented in \cite{wollic}.
In this paper, we significantly extend the results.
For example, we suggest a different set of axioms for the nondeterminism monad that is expressive enough to define various notions of nondeterministic completeness useful in practical applications.
We also expand several of the other sections and provide more details and additional examples.

\subsection{Overview and main contributions}
In this paper, we define an axiomatization of the reals in a constructive dependent type theory that allows program extraction.
Our axioms model real numbers in a conceptually similar way as mature implementations of exact real computation.
We further allow a certain amount of classical logic in proofs that do not have computational content. 
As we introduce several axioms in our theory, an important question is how we make sure that this kind of formalization is sound.
We argue that our formalization is sound by extending a realizability interpretation in the category of assemblies over Kleene's second algebra \cite{seely_1984,10.1007/BFb0022273}.
The theoretical framework for this kind of reasoning is given by computable analysis and the theory of representations \cite{kreitz1985theory,w00}.
We describe the necessary background in Section~\ref{s:background}.

As computable operations are necessarily continuous, nondeterminism is central in exact real computation.
We model nondeterministic computation as a monad that we introduce in Section~\ref{s:nondeterminism}.
For any type $X : \Type$,  $\mval X : \Type$ stands for the type of nondeterministic computations in $X$.
The idea is that when proving a theorem of the form $\all{x : \dR}P\,x \to \mval\some{y : \dR}Q\,x\,y$, a user automatically gets
an exact real computation program that for any input $x \in \IR$ such that $P(x)$ holds, nondeterministically computes $y \in \IR$ such that $Q(x, y)$ holds.

In Section~\ref{s:axiomatization} we define our real number axiomatization.
One of the defining features of exact real computation, making it more powerful than e.g. symbolic or algebraic computation, is the ability to construct real numbers by limits of certain user-defined sequences \cite{Bra03f,neumann2018topological}.
While axiomatizing the limit operation itself is not difficult, the situation becomes more involved for the case
when nondeterminism is present in the limit computation, and how to define such a nondeterministic limit operation has been under debate.
In the current work we specify \emph{nondeterministic dependent choice}, a simple and natural principle of the nondeterminism itself that makes the limit operation suggested in \cite{Mueller18} and some other forms of nondeterministic limits derivable.
We also suggest some simple extensions to complex numbers and general Euclidean spaces.

After defining this axiomatization, in Section~\ref{ss:axiomatization:realizability} we prove its soundness by extending the realizability interpretation from Section~\ref{ss:background:typetheory}.

As there are already many formulations of classical real numbers in proof assistants, from a practical point of view it would be desirable to transfer some of the proofs done over a classical axiomatization to our setting.
We address this in Section~\ref{s:relator} by defining a $\relator$ operation that lets us relate classical and constructive results in our type theory.

Finally, in Section~\ref{s:examples} we present several larger examples from computable analysis.

We furthermore implemented all of the theory (apart from the soundness in Section~\ref{ss:axiomatization:realizability}) in the Coq proof assistant.
We describe the details of the implementation in Section~\ref{s:implementation}.
The implementation allows to extract Haskell programs from our proofs using Coq's code extraction mechanism.
In the extracted programs, primitive operations on the reals are mapped to operations in the exact real computation framework AERN \cite{konecny2008aern} which is written and maintained by one of the authors.
We extracted code for several of the examples and measured its performance. We report that the extracted programs perform efficiently, having only a small overhead compared to hand-written implementations.

\section{Background}\label{s:background}
Before describing our axiomatization let us briefly summarize the necessary background from computable analysis and define the type theoretical setting in which we present our results.
\subsection{Computable analysis}\label{ss:computable analysis}
To compute over uncountable mathematical structures such as real numbers exactly, computable analysis takes \emph{assemblies} over Kleene's second algebra (assemblies for short in this paper) as the basic data type \cite{seely_1984,10.1007/BFb0022273,van2008realizability}.\footnote{Assemblies are generalizations of represented sets \cite{kreitz1985theory,w00} which are exactly the assemblies where the surjective relations are required to be partial surjective functions.
The terminology \emph{multi-representation} \cite{schroder2002effectivity} may be more familiar to some readers.}
An assembly is a pair of a set $X$ and a relation $\Vdash\, \subseteq \Baire\times X$, which is surjective in that $\forall x\in X.\ \exists \varphi.\ \varphi \Vdash x$ holds, in infix notation.
We call $\varphi \in \Baire$ a realizer of an abstract entity $x \in X$ if $\varphi \Vdash x$ holds.
We often write assemblies 
in boldface capitalized letters $\mathbf{X}, \mathbf{Y}, \cdots$. For an assembly $\mathbf{X}$, 
we write $|\mathbf{X}|$ for the underlying set of $\mathbf{X}$ and $\Vdash_\mathbf{X}$ for the 
realization relation of $\mathbf{X}$.
We often write $x \in \mathbf{X}$
for $x \in |\mathbf{X}|$ to simplify the presentation.

Given two assemblies $\mathbf{X}$ and $\mathbf{Y}$, a function $f : |\mathbf{X}| \to |\mathbf{Y}|$ 
is said to be computable if there is a computable (in the sense of type-2 computablity \cite{w00})
partial function $\tau :\subseteq \Baire \to \Baire$ that tracks $f$, i.e.\ 
for any $x \in \mathbf{X}$ and its realizer $\varphi$, $\tau(\varphi)$ is a realizer of $f(x)$.
Collecting assemblies as objects and computable functions between them as morphisms, 
the category of assemblies $\Casm$ forms a quasitopos.

An example of an assembly is the Cauchy assembly $\bR_\text{Cauchy}$ over the reals where $\varphi \in \Baire$ 
gets to be a realizer of $x \in \IR$ if and only if $\varphi$ encodes a sequence of rationals converging rapidly towards $x$.
That is, when we specify $\eta_\IQ : \IN \to \IQ$ to be the underlying encoding of rationals,
it holds that
\[
|x - \eta_\IQ(\varphi(n)) | \leq 2^{-n} \iff \varphi \Vdash_{\bR_\text{Cauchy}} x.
\]
Under the realization relation, the field arithmetic of reals is computable.
That means
the arithmetic operations get computed exactly without any rounding errors by the definition
of the realization relation.

An inevitable side-effect of this approach is partiality.
Consider comparing the order of two real numbers from the Cauchy assembly.
The naive procedure of iterating through the entries of the realizers of the two real numbers
will fail to terminate when the two real numbers are identical.
In fact, whichever realizability relation of reals we take, 
comparisons of real numbers are only partially computable \cite[Theorem~4.1.16]{w00}. 

A partial function $f :\subseteq X \to Y$ is called computable
with regards to the realization relations $\Vdash_X$ and $\Vdash_Y$ at its restricted domain $S \subseteq X$
if it, as a total function $f : S \to Y$, is computable with regards to $\Vdash_{S \subseteq X}$ and $\Vdash_Y$.
The realization relation $\Vdash_{S \subseteq X}$ on $S$ is a relation induced by the subset embedding 
from $\Vdash_X$ by 
\[
\varphi \Vdash_{S \subseteq X} x :\Leftrightarrow
\varphi \Vdash_X x
\]
With regards to the Cauchy assembly, the order comparison $<$ as a binary function to the Booleans is only partially computable when it is not defined at $\Delta_{\IR} := \{(x, x) \mid x \in \IR\}$.

Classically, partial functions $f :\subseteq X \to Y$ are often handled as a total function to a lifted codomain $f' : X \to Y\cup\{\bot\}$ where the new codomain contains an extra value $\bot$ which denotes 
that the function received a non-valid input. To be precise, let 
\[ f' := x \mapsto \begin{cases} f(x) & \text{if } x \in S \\ \bot &\text{otherwise,}\end{cases} \]
where $S \subseteq X$ is the actual domain of the partial function $f :\subseteq X \to Y$.
However, in order to do the same in our setting, to make the lifted codomain an assembly,
we need to consider the computational behavior of $\bot$ explicitly as we need to define the realizers of $\bot$.
It means, up to a way of defining the realizers of partialities, in this setting, partial functions can be classified by their computational behavior when they accept non-specified inputs \cite[Ch.~4]{longley1995realizability}.

Consider, for example, the problem of comparing the order of real numbers. 
Regarding it as a computable partial function 
that is not defined at $\Delta_\IR$, following the naive procedure of disginguishing the real numbers from their realizers, 
we can say that the computational behavior of the function when it receives inputs from $\Delta_\IR$ is diverging, e.g., non-termination.
Let \emph{Kleenean} $\bK$ be the assembly of  $\{\cff, \ctt, \bot\}$ where an infinite sequence of zeros realizes $\bot$, an infinite sequence that starts with $1$ after a finite prefix of zeros realizes $\cff$, and an infinite sequence that starts with $2$ after a finite prefix of zeros realizes $\ctt$ (see e.g. \cite[Example~3]{konevcny2020computable}). 
The assembly $\bK$ can be seen as a generalization of the Booleans
by adding an explicit state of divergence $\bot$.
Comparison in the Cauchy assembly is computable as a function to $\bK$:
\[
x < y := \begin{cases}
\ctt &\text{if } x < y, \\
\cff&\text{if } y < x,\\
\bot&\text{otherwise.}
\end{cases}
\]

That means, $\bK$ can be used to classify partial Boolean functions whose specification at non-specified inputs is to diverge (nontermination), including the order comparison of Cauchy reals.
As the name suggests, the assembly $\bK$ makes the Kleene three-valued logic operations, shown in Fig.~\ref{t:erc:kleene},  computable.
\begin{figure}
    \centering
    \begin{tabular}{c||c|c|c}
$\kland$ & $\ctt$ & $\cff$ & $\bot$ \\\hline\hline
$\ctt$   & $\ctt$ & $\cff$ & $\bot$ \\\hline
$\cff$   & $\cff$ & $\cff$ & $\cff$ \\\hline
$\bot$   & $\bot$ & $\cff$ & $\bot$ 
    \end{tabular}
    \hfill
    \begin{tabular}{c||c|c|c}
$\klor$ & $\ctt$ & $\cff$ & $\bot$ \\\hline\hline
$\ctt$   & $\ctt$ & $\ctt$ & $\ctt$ \\\hline
$\cff$   & $\ctt$ & $\cff$ & $\bot$ \\\hline
$\bot$   & $\ctt$ & $\bot$ & $\bot$ 
    \end{tabular}
    \hfill
    \begin{tabular}{c||c|c|c}
$\kneg$ & $\ctt$ & $\cff$ & $\bot$ \\\hline\hline
  & $\cff$ & $\ctt$ & $\bot$ 
    \end{tabular}
    \label{t:erc:kleene}
\caption{The Kleene three-valued logic operations
where $\bot$ represents undefinedness.}
\end{figure}

An assembly of reals satisfying the below computability conditions is called \emph{effective}.
\begin{fact}
\label{f:erc}
For real numbers, there is a unique assembly up to isomorphism that satisfies the computability-theoretic structure of real numbers  \cite{DBLP:journals/mlq/Hertling99}:
\begin{enumerate}
    \item The constants $0, 1\in \IR$ are computable.
    \item The arithmetical operations $+, -, \times$ are computable.
    \item The multiplicative inversion $/$ is computable as a partial function that is not defined at $0$.
    \item The order relation $<$ as a partial function that is not defined at $\{(x, x) \mid x \in \IR\}$ is computable. However, its lifted version in $\bK$ is computable.
    \item The limit operation defined at rapidly converging sequences is computable.
\end{enumerate}
\end{fact}
A concrete example of such an assembly of reals is, of course, the Cauchy reals. 
Hence, any effective axiomatization of real numbers has to capture this computational structure.

As the usual comparisons are partial, nondeterminism becomes essential in exact real computation  \cite{LUCKHARDT1977321}.
For two assemblies of $X$ and $Y$,
a multivalued function $f : X \mto Y$, which is basically a nonempty set-valued function, is computable if there is a computable function 
 that takes a realizer $\varphi$ of $x \in X$ and computes a realizer of any $y \in f(x)$. An example is the multivalued soft comparison \cite{BRATTKA1998490}:
\[
x <_k y = \{\ctt \mid x < y + 2^{k}\} \cup \{\cff \mid y < x + 2^{k}\}.
\]
The above total multivalued function approximates the order relation.
It is tracked by evaluating two partial comparisons $x < y + 2^{k}$ and $y < x + 2^{k}$ in parallel, returning $\ctt$ if $x < y + 2^{k} = \ctt$, and $\cff$ if $y < x + 2^{k} = \ctt$.
It is nondeterministic in the sense that for the same $x$ and $y$, but with different realizers, which of the tests terminates first may vary.
Exact real number computation software such as 
\cite{irram,konecny2008aern} further offer
operators like $\semselect :\subseteq \bK \times \bK \mto \mathbf{2}$ such that $\semselect(k_1, k_2) \ni \ctt$
iff $k_1 = \ctt$ and $\semselect(k_1, k_2) \ni \cff$
iff $k_2 = \ctt$ as a primitive operation for generating nondeterminism. Here, $\mathbf{2}$ is the canonical Boolean assembly.

\subsection{Type theory and realizability}
\label{ss:background:typetheory}
Type theory is a foundation of mathematics consisting of rules for forming new types and constructing new terms of the types.
In this paper, we work with a simple extensional dependent type theory which admits categorical semantics, i.e. a realizability interpretation,
in $\Casm$ \cite{streicher2012semantics,jacobs1999categorical,reus1999realizability}.
That is, each closed type $X$ is mapped to an assembly $\sem{X}$ and each 
closed term $x : X$ is mapped to a computable point / function $\sem{x} : \mathbf{1} \to \sem{X}$. 
Hence, the type theory can be used as a logical language of computable analysis.

We work with a type theory that provides the basic types:
$\mathsf{0}, \mathsf{1}, \mathsf{2}, \dN, \dZ$, an \`{a} la Russel universe of classical propositions $\Prop$, and an \`{a} la Russel universe of types $\Type$ (we omit the technical details necessary to avoid the usual paradoxes by e.g. imposing a hierarchy of types).
We interpret $\mathsf{0}$ as the initial assembly $\mathbf{0}$, the assembly of the empty set,
$\mathsf{1}$ as a final assembly, the trivial assembly of a singleton set $\mathbf{1}$,
$\mathsf{2}$ as the canonical assembly over $\{\ctt, \cff\}$,
$\mathsf{N}$ as the canonical assembly over $\IN$, 
and
$\mathsf{Z}$ as the canonical assembly over $\IZ$.
Here, triviality means that the underlying 
realization relation is the entire set.
We also assume sufficient enough operators over the
basic types and interpret them appropriately to the 
computable functions in $\Casm$ including 
primitive recursors over the finite and countable types.

The interpretations of type formations 
are done in a standard way.
For two types $X$ and $Y$, the (co)product type
($X+Y$) 
$X \times Y$
is interpreted as 
the (co)product of the interpretations of $X$ and $Y$. 
The function type $X \to Y$ is interpreted as 
the exponentiation of the interpretation of $X$ 
over the interpretation of $Y$.
Given a family of assemblies $\mathcal{F} : |\mathbf{X}| \to \Casm$ indexed by $|\mathbf{X}|$,
consider the two assemblies:
$\pi(\mathcal{F})$ of
the set of dependent functions $\{f : |\mathbf{X}| \to \bigcup_x |\mathcal{F}(x)| \mid 
\forall x.\; f(x) \in |\mathcal{F}(x)|\,\}$
and $\sigma(\mathcal{F})$ of 
the set of dependent pairs $\{(x, y) \mid x \in |\mathbf{X}| \land y \in |\mathcal{F}(x)|\,\}$.
The realization relations are defined by
$
\varphi \Vdash_{\pi(\mathcal{F})} 
f :\Leftrightarrow
\forall \psi \Vdash_\mathbf{X} x.\; 
\eta_\varphi(\psi) \Vdash_{\mathcal{F}(x)} f(x)
$
and
$
\langle \varphi_1, \varphi_2 \rangle \Vdash_{\sigma(\mathcal{F})} (x, y)
:\Leftrightarrow 
\varphi_1 \Vdash_\mathbf{X} x \land
\varphi_2 \Vdash_{\mathcal{F}(x)} y
$.
Here, $\eta_\phi$ denotes the underlying universal 
type-2 Turing machine indexed by $\phi\in \Baire$
and $\langle\varphi_1, \varphi_2\rangle \in \Baire$
denotes the Cantor paring of $\varphi_1, \varphi_2 \in \Baire$.
More details of the type-2 computability theory
can be found at \cite{w00}.

When we have a type family 
$P : X \to \Type$, or $Q : X \to \Prop$,
we interpret its $\allx$-type
$\all{x : X}P\,x$, or $\all{x : X}Q\,x$,
as the $\pi$-assembly of the family of assemblies 
expressed by $P$, or $Q$, respectively.
And, its $\somex$-type
$\some{x : X}P\,x$,
or
$\some{x : X}Q\,x$, is interpreted 
as the $\sigma$-assembly of the family of assemblies 
expressed by $P$, or $Q$, respectively.

We regard $\Prop$ as a universe of classical propositions. 
We interpret $\Prop$ as the trivial assembly of $\{\mathbf{0}, \mathbf{1} \}$
where $\mathbf{0}$ corresponds to $\tfalse : \Prop$
and $\mathbf{1}$ corresponds to $\ttrue : \Prop$.
We further suppose that $\Prop$ is a type universe that is closed under $\to, \times, \allx$. 
We often emphasize the property of $\Prop$ that it is a universe of (classical) propositions, 
by writing $\Rightarrow, \land, \forall$
instead of $\to, \times, \allx$.
However, note that $\Prop$ is not closed under 
$+$ and $\somex$. We assume the $\Prop$-level 
counterpart of the constructions $\lor$ and
$\csomex$
where $\neg\neg(X + Y) \simeq X \lor Y$
and $\neg\neg\some{x : X}Y \simeq \csome{x : X}Y$
get provable. Here, 
$\simeq$ denotes type-theoretic equivalent which 
in our case is simply category-theoretic isomorphism.

In order to make $\Prop$ classical, we assume the (classical) law of excluded middle 
\[ \all{P : \Prop} P \lor \neg P \label{a:TT1} \tag{Axiom TT1}\]
(where $\neg P :\equiv P \to \tfalse$), the (classical) propositional extensionality
\[ 
\all{P, Q : \Prop} (P\leftrightarrow Q) \to P = Q \label{a: TT2} \tag{Axiom TT2}
\](where $P \leftrightarrow Q:\equiv (P \to Q) \times (Q \to P)$),  
and the (classical) countable choice of the form
\begin{multline}
\all{X : \Type}\all{P : \dN \to X \to \Prop}
(\all{n : \dN} \csome{x : X}P\,n\,x) \to \\ \csome{f : \dN \to X}\all{n:\dN}P\,n\,(f\,n).
\label{a: TT3} \tag{Axiom TT3}
\end{multline}
Here, and on all later occasions when we introduce axioms, we sometimes just write the type of the axiom instead of the full axiom, by which we mean that the type is inhabited.
Note that the axioms are validated under the interpretation in $\Casm$.

We assume that the identity types $=$ are in $\Prop$
where we have extensionality, that we interpret $x = y$
as $\mathbf{0}$ when the interpretations of $x$ and $y$ are not identical
and as $\mathbf{1}$ when the interpretations are identical.

We also assume the general functional extensionality 
\begin{multline} \all{X : \Type}\all{P : X \to \Type}\all{f, g : \all{x : X}P\,x} \\
\big(\all{x : X}f\,x = g\,x\big) \to f = g \label{a: TT4} \tag{Axiom TT4}
\end{multline}
and the Markov principle
\begin{multline}
\all{f : \dN \to \Prop}\big(\all{n:\dN}(f\,n) + \neg(f\,n)\big) \\
\to (\csome{n : \dN} f\, n) \to \some{n : \dN} f\, n. \label{a: TT5} \tag{Axiom TT5}
\end{multline}

\section{Nondeterminism}\label{s:nondeterminism}
We propose an axiomatization of the nondeterminism considered in computable analysis as a monad $\mval : \Type \to \Type$ in our type theory
such that $\mval X : \Type$ stands for the type of nondeterministic computations in $X : \Type$.
The monad will be used intensively as, for example,
$f : A \to \mval B$ to denote a computable multivalued function 
$A \mto B$,
$\mval(A + B)$ to denote a nondeterministic decision procedure deciding either $A$ or $B$, 
and $\mval\some{x : A}B$ to denote a nondeterminstic existence of $x : A$ such that $B(x)$.

The aim is to achieve the interpretation that when we have a closed term of type $\all{x : A}P(x) \to \mval\some{y : B}Q\,x\,y$,
it automatically gives us a nondeterministic program 
that computes the partial multivalued function
\[
f :\subseteq \sem{A} \mto \sem{B}
\]
which is defined on $\{x \mid \sem{P}(x)\}$
satisfying $\forall y \in f(x).\; \sem{Q}(x, y)$.

The goal of this section is to propose a set of axioms that equips the type theory
with the monad. As in computable analysis basic nondeterminism is often introduced through Kleeneans, 
we first axiomatize the assembly $\dK$ in our type theory.
Then, believing that the type theory must be rich enough to do
classical reasoning on the nondeterministic values, we characterize 
the nondeterminism monad as a subobject of the classical non-empty power-set monad, 
which can be constructed within the type theory.
The latter part of the interpretation is dealt with in Section~\ref{ss:axiomatization:realizability}.


\subsection{Axiomatizing Kleeneans}
To axiomatize the Kleenans as described in Section~\ref{s:background}, we assume that there is a type $\dK : \Type$ of Kleeneans and that there are two \emph{distinct}  elements $\ctrue : \dK$
and $\cfalse : \dK$. 
That is, we introduce the following four axioms.
\begin{flalign}
& \label{a:K1} \kleene : \Type \tag{Axiom K1} \\ 
& \label{a:K2} \ctrue : \kleene \tag{Axiom K2}  \\
& \label{a:K3} \cfalse : \kleene \tag{Axiom K3}  \\
& \label{a:K4} \ctrue \neq \cfalse \tag{Axiom K4}  
\end{flalign}

Let us define the abbreviations
$
\upc{t} : \Prop :\equiv t = \ctrue
$ and
$ 
\downc{t} : \Prop :\equiv t = \cfalse.
$
For a Kleenean $t : \dK$, we further write $\defined{t}$ if $t$ is defined, i.e. 
for $\upc{t} \lor \downc{t} : \Prop$.
This allows us to introduce the standard logical operators (negation, conjunction and disjunction) on Kleeneans which follows the three-valued logic:
\begin{flalign}
&    \label{a:K5} \kneg : \kleene \to \kleene \tag{Axiom K5}\\
&    \label{a:K6} \klor : \kleene \to \kleene \to \kleene \tag{Axiom K6}\\
&    \label{a:K7} \kland : \kleene \to \kleene \to \kleene \tag{Axiom K7}\\
&    \label{a:K8} \upc{\kneg x} = \downc{x} \text{ and } \downc{\kneg x} = \upc{x} \tag{Axiom K8}\\
&    \label{a:K9} \upc{x \kland y} = (\upc{x} \land \upc{y}) \text{ and } \downc{x \kland y} = (\downc{x} \lor \downc{y}) \tag{Axiom K9} \\
&    \label{a:K10} \upc{x \klor y} = (\upc{x} \lor \upc{y}) \text{ and } \downc{x \klor y} = (\downc{x} \land \downc{y}) \tag{Axiom K10} 
\end{flalign}
If a Kleenean is defined, we assume that we can do branching on its value.
 \[
 \all{x : \dK} \defined{x} \to \coprodc{\upc{x}}{\downc{x}}\ . \tag{Axiom K11}\label{a:K11}
 \]

In many cases, we do not work directly with Kleeneans. 
Instead, we call a proposition $P : \Prop$ semi-decidable if
there is a Kleenean $t$ that identifies $P$:
\[
\semidec(P) :\equiv \some{t : \kleene} P = \upc{t}
\]
For example, in Section~\ref{s:axiomatization} we introduce the comparison operator $<\colon \dR \to \dR \to \Prop$ over the reals and assume 
$\all{x, y : \dR}\some{t : \kleene} \upc{t} = (x < y)$
to say that the operator is semi-decidable.
\subsection{Basic axiomatization of the nondeterminism monad}
Nondeterminism is expressed by a monad in our type theory such that when we have a type $X : \Type$, we automatically have its nondeterministic version $\mval X : \Type$.

We assume that there is a type constructor $\mval : \Type \to \Type$, a function lift $\lift^{\mval}$, a unit  $\unit^{\mval}$, and a multiplication $\mult^{\mval}$.
\begin{flalign}
  & \label{a:M1} \mval : \Type \to \Type  \tag{Axiom M1} \\
  & \label{a:M2} \unit^{\mval} : \all{X : \Type} X \to \mval\  X  \tag{Axiom M2} \\
  & \label{a:M3} \mult^{\mval} : \all{X : \Type} \mval\ (\mval\ X) \to \mval\ X  \tag{Axiom M3} \\
  & \label{a:M4} \lift^{\mval} : \all {X, Y : \Type} (X \to Y) \to (\mval\ X \to \mval\ Y)  \tag{Axiom M4}  
\end{flalign}
To simplify the presentation, let us write the type arguments to 
the unit, lift, and multiplication in subscript. For example, we write $\unit^\mval_X$ instead of $\unit^\mval X$. 

We add the following axioms to ensure that the type constructor, function lift, unit, and 
multiplication indeed form a monad. Namely, $\unit^{\mval}$ and $\mult^{\mval}$ are natural transformations and that the monad coherence conditions hold, i.e., 
\begin{flalign}
&  \label{a:M5} \lift^\mval_{X, Y} f\;(\unit^\mval_X\ x) = \unit^\mval_Y\ (f\ x)  \tag{Axiom M5} \\
  & \label{a:M6} \mult^\mval_Y\ (\lift^\mval_{\mval X,\mval Y}\  (\lift^\mval_{X, Y} f)\ y) = (\lift^\mval_{X, Y} f)\  (\mult^\mval_X y)  \tag{Axiom M6} \\
  & \label{a:M7} \mult^\mval_X\ (\unit^\mval_{\mval X}\ y) = y  \tag{Axiom M7} \\
  & \label{a:M8}  \mult^\mval_X\ (\lift^\mval_{X , \mval X}\ \unit^\mval_X \ y) = y  \tag{Axiom M8} \\
  & \label{a:M9}  \mult^\mval_X\ (\mult^\mval_{\mval X}\ z) = \mult^\mval_X\ (\lift^\mval_{\mval (\mval X), \mval X} \ \mult^\mval_X \ z)  \tag{Axiom M9} 
\end{flalign}
for all $X, Y : \Type$, $f: X \to Y$, $x : X$, $y :  \mval\ (\mval\ X)$ and $z :  \mval\ (\mval\ (\mval\ X))$.

A term of the nondeterministic type $\mval X$ is regarded as the result of a nondeterministic computation in $X$.
Under a belief that our type theory has to be effective enough to do classical reasoning 
on the possible return values of the nondeterministic computation,
we suggest a characterization of the monad by relating it with the classical 
non-empty power-set monad that we can construct within the type theory:
\[
\cval X :\equiv \some{S : X \to \Prop} \csome{x : X}S\,x
\]
We can confirm that the coherence conditions for the monad hold with function lift, unit, and multiplication:
\begin{align*}
&\lift^{\cval}_{X, Y}\,f
&:\equiv& \quad\lam{(S, -)} \big(\lam{y:Y}\csome{z:X}(y = f\,z) \times (S\,z), -\big)
\\
&\unit^{\cval}_X\,x &:\equiv&\quad (\lam{y : X}x = y, -) \\
&\mult^{\cval}_X\,x &:\equiv&\quad \big(\lam{y : X} \some{z : \cval X} (\pi_1\,z\,y) \times (\pi_1\,x\,z), -\big)
\end{align*}
Here, $\pi_1$ is the first projection of $\somex$-types and the occurrences of $-$ represent some classical proof terms.

In order to relate it with the classical non-empty power-set monad, 
we assume that there is a submonoidal natural transformation 
\[
\pic : \all{X : \Type} \mval X \to \cval X.
\label{a:M10} \tag{Axiom M10}
\]
It is a submonoidal natural transformation in that (i) it is a natural transformation
\[
\all{X, Y : \Type}
\all{f : X \to Y}
(\lift^\cval_{X, Y}f) \circ \pic_X
= \pic_Y\circ(\lift^\mval_{X, Y}f )
, 
\label{a:M11} \tag{Axiom M11}
\]
(ii) for any $X : \Type$, $\pic_X$ is monic
\[
\all{x :X}
\all{y :X}
\pic_X x = \pic_X y \to
x = y,
\label{a:M12} \tag{Axiom M12}
\]
and (iii) the coherence conditions
which on the unit is 
\[\pic_X\circ\unit^{\mval}_X = \unit^{\cval}_X
\label{a:M13} \tag{Axiom M13}
\] 
and on the multiplication is 
\[
\pic_X \circ \mult^{\mval}_X = \mult^{\cval}_X\circ \pic_{\cval X} \circ (\lift^{\mval} \pic_X)
\label{a:M14} \tag{Axiom M14}
\]
hold.
In other words, thinking of it in the setting of category theory, we assume that our monad $\mval$ is a submonoidal object of $\cval$ in the category of endofunctors. 

We further characterize the nondeterminism by that the classically lifted picture $\lift^\cval\,\pic_X : \cval (\mval X) \to \cval (\cval X)$ is a natural isomorphism that admits an inverse:
\[
\all{X : \Type}\textsf{is\_equiv}\ (\lift^\cval_{\mval X, \cval X} \pic_X)
\label{a:M15} \tag{Axiom M15}
\]
The following diagram shows the relation between the two monads.
\begin{center}
\begin{tikzcd}
  X \arrow[rr, "\unit_X^\mval"] \arrow[rrdd, "\unit_X^\cval"'] && \mval X \arrow[dd, "\pic_X"] \arrow[rr,"\unit_{\mval X}^\cval"]&& \cval (\mval X) \arrow[dd, "\lift^{\cval}\,\pic_X", "\sim" labl]\\ &&&&\\
                                     && \cval X                       \arrow[rr,"\unit_{\cval X}^{\cval}"']&& \cval (\cval X)
 \end{tikzcd}
\end{center}

We call the natural transformation ``picture'' because we regard $\pic_X\, x$, when $x : \mval X$ is a nondeterministic element,
as a classical picture showing the elements $x$ represents
such that we can do reasoning on them. 
However, if the purpose is to refer to the elements of $x : \mval X$, 
we often do not need the second part of $\pic_X\, x$,
which is about the non-emptiness of $x$. 
Hence, let us make the definition 
\[
\picc_X : \mval X \to (X \to \Prop) :\equiv 
\lam{x : \mval X}\pi_1 (\pic_X x)
\]
so that we can conveniently express 
$\picc_X \ x \ y : \Prop$ to say that $y : X$ is a possible outcome of $x : \mval X$.

The last building block in relating the two monads is a destruction method:
\[
\all{X : \Type} \mval X\to  \mval \some{y : X} \picc_X\,x\,y.
\label{a:M16} \tag{Axiom M16}
\]
Namely, when we have a nondeterministic object $x$, we can nondeterministically get a pair $(y, t)$ where $y : X$ and $t$ is a reason why $y$ can be nondeterministically obtained from $x$.

We introduce some additional axioms to characterize nondeterministic computation.
For any two semi-decidable decisions $x, y : \dK$, if promised that either of $x$ or $y$ holds classically, we can nondeterministically decide whether $x$ holds or $y$ holds:
\[
\select :
\all{x, y : \kleene} (\upc{x} \lor \upc{y}) \to \mval \big(\coprodp{\upc{x}}{\upc{y}}\big)\,. \label{a:M17} \tag{Axiom M17}
\]
If a type $X$ is subsingleton, we can eliminate the nondeterminism on $\mval X$:
\[
\singletonextract :\all{X : \Type} (\all{x, y : X} x = y) \to 
\mathsf{is\_equiv}_{X, \mval X}\ \unit^\mval_X. \label{a:M18} \tag{Axiom M18}
\]
Here, $\mathsf{is\_equiv}_{X, Y}f$,
a type saying that two types $X , Y $ are
equivalent by $f : X \to Y$, is defined by
$\some{g : Y \to X}(\all{x : X}g\,(f\,x) = x) \land
\all{y : Y}f\,(g\,y) = y$
as our type theory is extensional.

\begin{example}
For any proposition $P$, suppose both $\semidec(P)$ and $\semidec(\neg P)$ hold.
As $P \lor \neg P$ holds by the classical law of excluded middle, 
we have $\mval(\coprodp{P}{\neg P})$ by applying $\select$. 
As it is provable that $\coprodp{P}{\neg P}$ is subsingleton,
using $\singletonextract$, we have $\coprodp{P}{\neg P}$, the decidability of the proposition $P$.
\end{example}

\begin{example}
\label{ex:1002}
An example of using $\pic$ is when we analyze the values of nondeterministically lifted functions.
See the commutative diagram:
\begin{center}
\begin{tikzcd}
 & X \arrow[lddd]\arrow[rr, "f"]\arrow[dd] & & Y \arrow[dd]\arrow[lddd] \\
 & & & \\
 & \mval X \arrow[rr,"\lift^\mval_{X, Y}f"] \arrow[ld,"\pic_X"'] & & \mval Y \arrow[ld,"\pic_Y"]\\
 \cval X \arrow[rr,"\lift^\cval_{X, Y}f"'] & & \cval Y & 
 \end{tikzcd}
\end{center}
When we have a lifted function $\lift^\mval_{X, Y} f$
and a nondeterministic object $x : \mval X$,
we see the nondeterministic output
$\lift^\mval_{X, Y} f\ x$ by taking its picture
$\pic_Y\ (\lift^\mval_{X, Y} f\ x)$.
Directly from the assumptions of $\pic$,
propositional extensionality, and 
functional extensionality, we can easily 
show that
\[
\picc_Y\ (\lift^\mval_{X, Y} f \ x)
= \lam{z : Y} \csome{y : X} \picc_X x \ y \land z = f \ y
\]
holds.
In words, $z : Y$ is a possible output of 
$(\lift^\mval_{X, Y} f\ x)$
if and only if
there classically is $y : X$ such that
$y$ is a possible output of $x : \mval X$
and $z$ is $f\ y$.

The equation can be used in the following implicational forms:
\[
\picc_X x \ y \to \picc_Y (\lift^\mval_{X, Y} f \ x)\ (f\ y)
\]
\[
\picc_Y (\lift^\mval_{X, Y} f \ x)\ z \to
\csome{y : X} \picc_X x \ y \land z = f \ y
\]

\end{example}


\subsection{Nondeterministic dependent choice}\label{s:nondet-dep-choice}
Suppose any sequence of types $P : \dN \to \Type$ and a nondeterministic procedure that runs through the 
types $f : \all{n : \dN} (P\,n) \to \mval(P\,(n+1))$. 
We can think of a procedure of repeatedly and indefinitely applying the nondeterministic procedure: e.g.,
$f_n(\cdots f_2(f_1(f_0\,x_0))\cdots)$ where $x_0 : P\,0$. Though the expression is not well-typed,
intuitively, from the computational point of view, when we apply it repeatedly, we nondeterministically get a sequence that selects through $P\,n$. 
Starting from $x_0$, we get nondeterministically $x_1 : P\,0$ from $f\, 0\,x_0 : \mval\, (P\,0)$. Then, according to the nondeterministic choice $x_1 : P\,0$ amongst $f\, 0\,x_0 : \mval\, (P\,0)$, we again get nondeterministically $x_2 : P\,1$ from $f\, 1\, x_1 : \mval\, (P\,1)$. Repeating this forever, 
we get a specific (nondeterministic) sequence where each entry depends on the nondeterministic choices that have been made in the previous entries.

One may think the primitive recursion of natural numbers, which already exists in the base type theory,
and which is meant to express repeated applications, will do the job.
The primitive recursion of natural numbers applied to the type family $P$ is of type
\[
\mval\, (P\,0) \to (\all{n:\dN}\mval \,(P\,n) \to \mval \,(P\,(n + 1))) \to 
\all{n : \dN}\mval\,(P\,n).
\]
Given $f : \all{n : \dN} (P\,n) \to \mval(P\,(n+1))$ and $x_0 : P\,0$, applying the
recursion on $\unit^\mval_{P\,0} x_0$ and $\lam{n : \dN}\lam{x : \mval (P\,n)}\mult^\mval(\lift^\mval (f\, n)\,x)$ denotes exactly
applying $f$ repeatedly on $x_0$.
However, the result of the application does not preserve any information on 
the dependency between the sequential nondeterministic choices as we can see that the result is of type $\all{n : \dN}\mval (P\,n)$.

For example, let us consider $P\,n :\equiv \dN$  and  
\[f\,n\,x :\equiv \begin{cases}0 \text{ or } 1 &\text{if }n = 0, \\ x &\text{otherwise.}\end{cases} \]
When we repeatedly apply the procedure on $0$, we expect to have one of the two sequences
$0,0,0,0,0,\cdots$ or $0,1,1,1,1,\cdots$ nondeterministically. However, when we apply 
the primitive recursion, all we can get is the sequence of the nondeterministic real numbers
$0, (0\text{ or } 1), (0\text{ or } 1),\cdots$ which is less informative, forgetting all the information about the dependencies
that $f$ creates.
Hence, we need a separate and more expressive principle but with computational behavior identical to primitive recursion.

Suppose any sequence of types $P : \dN \to \Type$ and a sequence of classical binary relations
$R :\all{n : \dN} P\,n \to P\,(n+1) \to \Prop$. The binary relation is where the 
dependencies between sequential choices are encoded.
For the above example, $R\,n\,x\,y$ can be set to $n > 0 \to x = y$.
We call a function of type
\[
\all{n : \dN}\all{x : P\,n}\mval\some{y : P\,(n+1)}R\,n\,x\,y
\]
an $\mval$-\emph{trace} of $R$. Note that admitting a trace automatically ensures that $R$ is a (classically) entire relation.

The nondeterministic dependent choice ($\mval$-dependent choice for short) 
says that for any $\mval$-trace of $R$, there is a term of type
\[
\mval\some{g : \all{n : \dN}P\,n}\all{m : \dN}R\,m\,(g\,m)\,(g\,(m + 1))
\]
satisfying a coherence condition that will be described below. 
In words: From a trace of $R$, 
we can nondeterministically get a sequence $g$ that runs through $R$.

Given any $\mval$-trace $f$ of $R$, 
now there are two different ways of constructing a term of type
$\all{n : \dN}\mval(P\,n)$, forgetting
the information on the dependencies.
The first is to naively apply the primitive recursion on $f$ as described in the beginning of this subsection.
The second, is to apply the following operation
\[
\mathsf{to\_fiber}^\mval(g : \mval\all{n : \dN}P\,n) :\equiv \lam{n : \dN} \big(\lift^\mval (\lam{h : \all{n: \dN}P\, n} h\, n)\, g\big)
\]
on the $\mval$-lifted first projection 
of the $\mval$-dependent choice.
The coherence condition states that the two operations of forgetting the information on the paths are identical (c.f.\ Figure~\ref{fig:dchoice}).

\begin{figure}
    \centering
    \includegraphics[width=.8\hsize]{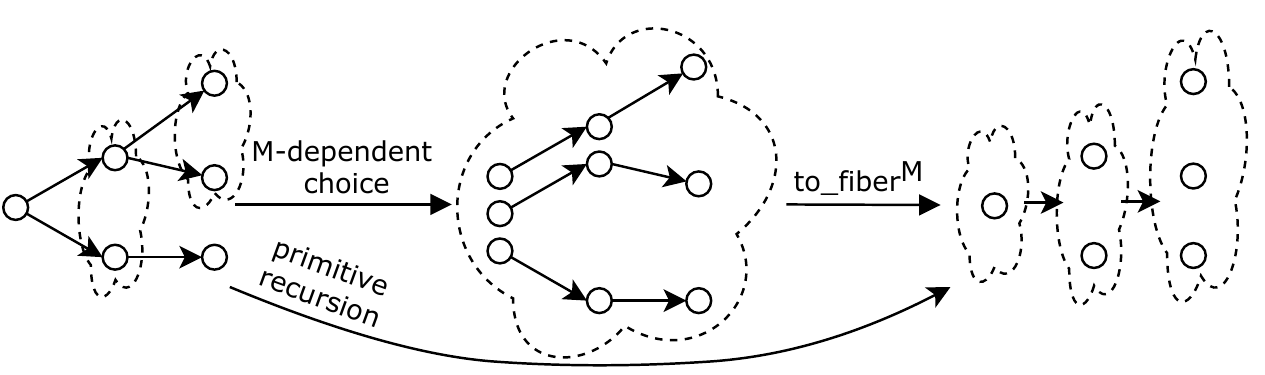}
    \caption{Intuitive picture on the coherence condition for the $\mval$-dependent choice}
    \label{fig:dchoice}
\end{figure}

\[
\textit{our type theory admits $\mval$-dependent choice.}
\label{a:M19}
\tag{Axiom M19}
\]
\begin{remark}
The name nondeterministic dependent choice comes from the observation
that when repeating the above with the double negation
monad or the propositional truncation monad (assuming they are provided by the type theory),  the principle becomes the classical dependent choice and intuitionistic dependent choice, respectively.
\end{remark}

In an earlier version \cite{wollic}, we axiomatized that there is a term constant $\countablelift$
such that for any $P : \dN \to \Type$, it holds that
$
\countablelift\,P: (\all{n :\dN}\mval\,(P\,n)) \to \mval\all{n : \dN} P\,n $ is a section of $\mathsf{to\_fiber}^\mval$.
In other words, for any 
$f : \all{n :\dN}\mval\,(P\,n)$, it holds that
$\mathsf{to\_fiber}\,(\countablelift\,f) = f$.
From a computational point of view, it says that when we have a sequence of nondeterministic computations, we can nondeterministically choose one sequence of (deterministic) computations. 
The countable lift is derivable from nondetermistic dependent choice.

\begin{lemma}
The countable lift property $\countablelift$ is derivable.
\end{lemma}
\begin{proof}
Suppose any sequence of types $P : \dN \to \Type$ and a sequence of nondeterministic objects
$f : \all{n : \dN}\mval(P\, n)$. 
Then, we can set $R\,n\,x\,y :\equiv \ttrue$ and verify that the function
\[
\lam{n : \dN}\lam{x : P\,n} (f\,(n+1), \star)
\]
is a trace where $\star : \ttrue$. 
Hence, applying the $\mval$-lifted dependent choice, we can get a term of type $\mval\all{n : \dN}P\, n$.
\end{proof}

\subsection{Nondeterministic logic}
When we have a nondeterministic object $x : \mval X$, regarding it as the result of some nondeterministic computation, 
it is desirable to analyze properties of the possible outcomes of the nondeterministic computation $x$.
For example, when we have a nondeterministic real number $x : \mval \dR$, we might want to know 
if every possible outcomes of $x$ is non-zero or if there is at least one possible outcome which is zero.
In general, for a classical predicate $P : X \to \Prop$, for a nondeterministic object $x : \mval X$, 
it is desirable to have a way to express a proposition such as \emph{all possible outcomes $y : X$ of $x$ satisfies $P\,y$}
or \emph{some possible outcomes $y : X$ of $x$ satisfies $P\,y$}.

The most simple way to achieve this is to use the natural transformation $\pic_X : \mval X \to \cval X$
so that we can work with the classical description of $x : \mval X$.
For example, when we have $P : X \to \Prop$, the proposition 
that all outcomes of $x$ satisfying $P$ can be described by 
type $\all{y : X}\picc_X\,x\,y \to P\,y$
and the proposition that some outcomes of $x$ satisfying $P$ can be 
described by type $\csome{y : X}\picc_X\,x\,y \times P\,y$.
However, the characterization of the nondeterminism allows another formalization.
When we have $P : X \to \Prop$, we can naturally lift it to 
$\lift^\mval P : \mval X \to \mval \Prop$. Hence, if 
there are retractions of $\unit^\mval_\Prop : \Prop \to \mval\Prop$,
we can postcompose them to get naturally defined predicates of type
$\mval X \to \Prop$.

\begin{lemma}
Define the two operators:
\begin{align*}
\land^\mval &: \mval \Prop \to \Prop :\equiv \lam{x : \mval \Prop} x = \unit^\mval \ctrue\\
\lor^\mval &: \mval \Prop \to \Prop :\equiv \lam{x : \mval \Prop} \neg (x = \unit^\mval \cfalse)
\end{align*}
They are retractions of $\unit^\mval_\Prop : \Prop \to \mval \Prop$.
In other words, for any $P : \Prop$, it holds that
$\land^\mval(\unit^\mval_\Prop\,P) = P$ 
and $\lor^\mval(\unit^\mval_\Prop\,P) = P$ hold.
\end{lemma}
\begin{proof}
To prove that $\land^\mval$ is a retract, 
we need to show that $\land^\mval (\unit^\mval_\Prop P) = P$ 
holds for any $P$.
As the target, which is an identity type, is in $\Prop$, we 
can do branching on $P \lor \neg P$ using \ref{a:TT1}. Further using the classical propositional extensionality (\ref{a: TT2}), 
we only need to show that $\land^\mval (\unit^\mval_\Prop \ttrue) = \ttrue$
and $\land^\mval (\unit^\mval_\Prop \tfalse) = \tfalse$ hold.

Unfolding the definition of $\land^\mval$, it means to prove 
\[
(\unit^\mval_\Prop \ttrue = \unit^\mval_\Prop \ttrue) = \ttrue 
\]
and
\[
(\unit^\mval_\Prop \tfalse = \unit^\mval_\Prop \ttrue) = \tfalse 
\]
which are obvious from the classical propositional extensionality
and the fact that $\unit^\mval$ is monic.
The fact that $\unit^\mval$ is monic is direct from the fact that 
$\pic$ and $\unit^\cval$ are monic.

The other case of $\lor^\mval$ can be proven similarly.

\end{proof}

Let us for a moment depart the type theory and think in the classical setting where we consider $\Prop$ as a set $\{\top, \bot\}$. 
And, 
$\mval\Prop$ is regarded as a set $\{\{\top, \bot\}, \{\top\}, \{\bot\}\}$.
Hence, classically, there can be exactly two retractions of the unit which is decided by whether
we map $\{\top, \bot\}$ to $\top$ or to $\bot$. 
It is not too difficult to see $\land^\mval$ is a retraction that is mapping $\{\top, \bot\}$ to $\bot$
and $\lor^\mval$ is a retraction that is mapping $\{\top, \bot\}$ to $\top$. 

\begin{definition}
For any type $X$, a classical predicate $P : X \to \Prop$, and a nondeterministic object $x : \mval X$, let us define
\[
\forall^\mval x\, P  :\equiv \land^\mval \circ \lift^\mval P 
\quad\text{and}\quad
\exists^\mval x\,P  :\equiv \lor^\mval \circ \lift^\mval P .
\]
Let us write $\forall^\mval y : x.\; P(y)$ for $\forall^\mval x\,(\lam{y :X} P(y))$ 
and $\exists^\mval y : x.\; P(y)$ for $\exists^\mval x\,(\lam{y :X} P(y))$.
And, for any $x : \mval X$ and $y : X$, define
\[
y \in^\mval x :\equiv \exists^\mval z : x.\;y = z.
\]
\end{definition}

Then, we can prove in the type theory that the axiomatization of $\pic$ works as intended by proving the following lemma.
\begin{lemma}
Within the type theory, for any type $X$, a classical predicate $P : X \to \Prop$, and a nondeterministic object $x : \mval X$, 
the following equations hold:
\[
\forall^\mval x\, P = \all{y : X}\picc_X \,x\,y \to P\,y
\;\;\text{and}\;\;
\exists^\mval x\, P = \csome{y : X}\picc_X \,x\,y \land P\,x.
\]
And, for any $y : X$, it holds that
\[
y \in^\mval x = \picc_X \,x\,y.
\]
\end{lemma}
\begin{proof}
For the first equation, 
using the classical propositional extensionality, 
we only need to prove that 
$
(\forall^\mval x\, P) \to (\all{y : X}\picc_X \,x\,y \to P\,y)$
and
$
(\all{y : X}\picc_X \,x\,y \to P\,y) \to (\forall^\mval x\, P)$ hold.
Unfolding the definitions, the implications become
\begin{equation}\label{eqn:1000}
(\lift^\mval_{X, \Prop} P\ x = \unit^\mval_\Prop \ttrue) \to (\all{y : X}\picc_X \,x\,y \to P\,y)
\end{equation}
and
\begin{equation}\label{eqn:1001}
 (\all{y : X}\picc_X \,x\,y \to P\,y) \to
 (\lift^\mval_{X, \Prop} P\ x = \unit^\mval_\Prop \ttrue).
\end{equation}

For Equation~\ref{eqn:1000}, 
we need to prove that $P\ y$ holds 
with the assumptions 
$\lift^\mval_{X, \Prop} P\ x = \unit^\mval_\Prop \ttrue$, $y : X$, and $\picc_X x \ y$. 
From Example~\ref{ex:1002}, 
since we have $\picc_X x \ y$ as an assumption, 
we can derive
$\picc_\Prop (\lift_{X, \Prop}^\mval P\ x)\ (P\ y)$.
Thanks to the identity type
$\lift^\mval_{X, \Prop} P\ x = \unit^\mval_\Prop \ttrue$, we get 
\begin{equation}\label{eqn:1002}
\picc_\Prop (\unit^\mval_\Prop \ttrue)\ (P\ y).
\end{equation}

Therefore, by \ref{a:M13}
on Equation~\ref{eqn:1002},
we get $\unit_\Prop^\cval \ttrue \ (P\ y)$ .
Unfolding the definition of $\unit^\cval$, we can get $P\ y = \ttrue$
which is enough to prove the goal $P\ y$.

\vspace{1em}

For Equation~\ref{eqn:1001}, we need to prove the identity
$\lift^\mval_{X, \Prop} P\ x = \unit^\mval_\Prop \ttrue$
assuming $\all{y : X}\picc_X \,x\,y \to P\,y$.
Applying the fact that $\picc$ is monic, we can simplify the goal to
\[
\picc_\Prop (\lift^\mval_{X, \Prop} P\ x) = \picc_\Prop (\unit^\mval_\Prop \ttrue)
\]
where the right hand side can be reduced to $\lam{y : \Prop}y = \ttrue$
by \ref{a:M13}.
Hence, applying the functional extensionality and classical propositional extensionality, the goal gets divided into the two sub-goals:
\begin{equation}
\label{eqn:1003}
(\picc_\Prop (\lift^\mval_{X, \Prop} P\ x) \ y) \to
y = \ttrue
\end{equation}
and
\begin{equation}
\label{eqn:1004}
\picc_\Prop (\lift^\mval_{X, \Prop} P\ x) \ \ttrue . 
\end{equation}

For Equation~\ref{eqn:1003},
from Example~\ref{ex:1002}, we have
the equation $\picc_\Prop (\lift^\mval_{X, \Prop} P\ x) \ y) =
\csome{z : X} \picc_X x\ z \land y = P\ z $.
Hence, we need to derive $y = \ttrue$
from $\csome{z : X} \picc_X x\ z \land y = P\ z $
and $\all{y : X}\picc_X \,x\,y \to P\,y$.

Since the goal is in $\Prop$, we can instantiate
$z : X$ such that $\picc_X x\ z$ and  $y = P\ z$ hold.
Therefore, from the initial assumption $\all{y : X}\picc_X \,x\,y \to P\,y$
we have $P\ z$.
From classical propositional extensionality, we can prove $P\ z = \ttrue$.
Hence, we can derive $y = \ttrue$.

For Equation~\ref{eqn:1004},
from Example~\ref{ex:1002}, we have
the equation $\picc_\Prop (\lift^\mval_{X, \Prop} P\ x) \ \ttrue) =
\csome{z : X} \picc_X x\ z \land \ttrue = P\ z $.
Hence, the goal can be rewritten as
\[
\csome{z : X} \picc_X x\ z \land \ttrue = P\ z 
\]
Since the goal, which is in $\Prop$ is provably subsingleton, 
we can reduce the goal to 
\[
\mval\csome{z : X} \picc_X x\ z \land \ttrue = P\ z 
\]
using \ref{a:M18}. 

Using \ref{a:M16}, from the assumption that we have $x : \mval X$,
we have $\mval\some{z : X} \picc_X x \ z$.
Since the goal is $\mval$-lifted, we can apply $\lift^\mval$
to get
$\some{z : X} \picc_X x \ z$.
Hence, we can get $z : X$ such that 
$ \picc_X x \ z$ holds. 
Now, applying it in the initial assumption $\all{y : X}\picc_X \,x\,y \to P\,y$,
we get $P \ z$. Hence, applying $\unit^\mval$ with the $z$ 
being the witness of the $\csomex$ type, we get 
\[
\mval\csome{z : X} \picc_X x\ z \land \ttrue = P\ z .
\]

\vspace{1em}

Other equations can be proven similarly.

\end{proof}

In practice, when using our axiomatization, 
we often get an object of type $\mval\some{x : X}P\,x$. 
For example, by constructing a term $f : \coprodp{P}{Q} \to \some{x : X}P\,x$
using the case distinction on $P$ or $Q$, when we lift it, we get 
$\lift^\mval f : \mval(\coprodp{P}{Q}) \to \mval\some{x : X}P\,x$. Then, 
precomposing an assumption that $P$ and $Q$ are nondeterministically choosable, 
we get a term of type $\mval\some{x : X}P\,x$ meaning that 
we can nondeterministically obtain $x : X$ such that $P\, x$ holds.
Since now we have a tool to deal with each entries of a nondeterministic object,
we can come up with the following lemma 
that splits $\mval$-lifted $\somex$-types:
\begin{lemma}
For any type $X : \Type$ and a classical predicate $P : X \to \Prop$, it holds that
\[
\big(\some{x : \mval X} \forall^\mval y : x.\,P\,y\big) \leftrightarrow 
\mval \some{x : X} P\,x
\]
\end{lemma}

In words, when we have a nondeterministic existence $\mval\some{x : X} P\,x$, we 
can split it to obtain a deterministic existence of 
a nondeterministic object $x : \mval X$
such that every entry of it satisfies $P$, and vice versa. 

\section{Real numbers and Limits}
\label{s:axiomatization}
Let us now introduce the basic axioms that we use for the real number type.
A defining feature of exact real computation is the ability to compute certain limits of user-defined sequences.
Its counterpart in the axiomatization of real numbers is the principle of metric completeness.

There are three distinct cases where we need to compute limits:
\begin{enumerate}[label=(\arabic*)]
\item when a deterministic sequence of real numbers converge to a deterministic point, 
\item when a sequence of nondeterministic real numbers converge to a deterministic point, and 
\item when a sequence of nondeterministic real numbers converge to a nondeterministic point.
\end{enumerate}
We deal with each of them separately in Sections~\ref{ss:deterministic limits} - \ref{ss:nondet limits}.

While the first one, i.e., how the limit computation should be axiomatized in the logical language is well-established \cite{Spe49,bishop1967foundations,BRIDGES199995},
it has been under debate how to deal with the case when nondeterminism is involved in the limit computation \cite{Mueller18,Konecny18}.
Note that this situation occurs quite naturally even for simple operations such as computing square roots of complex numbers (c.f. Section~\ref{ss:complex-root}.
Early versions of iRRAM therefore already provided a simple nondeterministic limit operation as primitive \cite[\S~10.3]{irram} which the authors of the software recently suggested to replace by a more generic operation for nondeterministic limits \cite{Mueller18}. 
In Section~\ref{s:nondet-dep-choice} we define the \emph{nondeterministic dependent choice}, a simple and natural principle of nondeterminism which we apply in Section~\ref{ss:nondet limits} to show that the nondetermistic limit operation can be derived in our theory and therefore does not need to be introduced as a primitive operation.

Finally, we propose some simple extensions of the theory of real numbers to complex numbers and general euclidean spaces.
\subsection{Basic axiomatization of real numbers}
We assume real numbers by declaring that there is a type for real numbers containing two distinct constants $0$ and $1$ and the standard arithmetical operators.
\begin{flalign}
& \dR : \Type \label{a:R1}\tag{Axiom R1}   \\
& \label{a:R2} 0 : \dR \tag{Axiom R2}  \\
& \label{a:R3} 1 : \dR \tag{Axiom R3}  \\
& \label{a:R5} + : \dR \to \dR \to \dR \tag{Axiom R4}  \\
& \label{a:R6} \times : \dR \to \dR \to \dR \tag{Axiom R5}  \\
& \label{a:R7} - : \dR \to \dR \tag{Axiom R6}  \\
& \label{a:R8} / : \all{x : \dR} x \neq 0 \to \dR \tag{Axiom R7} 
\end{flalign}
 
We further assume the existence of a semi-decidable comparison operator.
\begin{flalign}
& \label{a:R9} < : \dR \to \dR \to \Prop \tag{Axiom R8} \\
& \label{a:R10} \all{x, y : \dR}\semidec(x < y) \tag{Axiom R9} 
\end{flalign}
That is, for any two real numbers its order (as a classical proposition) is semi-decidable.

We assume the properties of the structure classically in a safe way
that does not damage constructivity.
\begin{flalign}
 &   \label{a:R11} \all{x, y:\dR} x + y = y + x \tag{Axiom R10} \\
&   \label{a:R12} \all{x, y,z:\dR} (x + y) + z = x + (y + z) \tag{Axiom R11} \\
& \label{a:R13} \all{x:\dR} x + - x = 0\tag{Axiom R12} \\
& \label{a:R14} \all{x:\dR} 0 + x = x\tag{Axiom R13} \\
&  \label{a:R15} \all{x, y:\dR}x \times y = y \times x\tag{Axiom R14} \\
&  \label{a:R16} \all{x, y, z : \dR} (x \times y) \times z = x \times (y \times z)\tag{Axiom R15} \\
& \label{a:R17} \all{x:\dR}\all{p : x \neq 0} (/\ x\ p) \times x = 1\tag{Axiom R16} \\
& \label{a:R18} \all{x:\dR} 1 \times x= x\tag{Axiom R17} \\
& \label{a:R19}  \all{x,y,z:\dR}x \times (y + z) = x \times y + x \times z\tag{Axiom R18} \\
& \label{a:R4} 1 \neq 0 \tag{Axiom R19}  \\
& \label{a:R20} 1 > 0\tag{Axiom R20} \\
& \label{a:R21} \all{ x, y :\dR }x < y \lor x = y \lor x > y\tag{Axiom R21} \\
& \label{a:R22} \all{x,y : \dR }x < y \to \neg (y < x)\tag{Axiom R22} \\
& \label{a:R23} \all{ x,z,y :\dR }x < y \to y < z \to x < z\tag{Axiom R23} \\
& \label{a:R24} \all{ x, y, z :\dR}y < z \to x + y < x + z\tag{Axiom R24} \\
& \label{a:R25} \all{x, y, z : \dR}0 < x \to y < z \to  x\times y  < x \times z\tag{Axiom R25}  \\
& \all{x : \dR} x > 0 \to \csome{n : \dN} 2^{-n} < x \label{a:R27} \tag{Axiom R26}
\end{flalign}

Note, that for example trichotomy is only assumed classically and  an inhabitant of the type
$\all{x, y : \dR} (x < y) + (x = y) + (y < x)$ is not posed anywhere.
\begin{example}
\label{exa:split}
As real comparison is only semi-decidable, in exact real computation it is often replaced by a so-called $\epsilon$-test which roughly says, for any real numbers $x, y, \epsilon$, when $\epsilon$ is positive, 
we can nondeterminstically decide if $x < y + \epsilon$ or $y < x + \epsilon$.
This corresponds to the multivalued version of the approximate splitting lemma \cite[Lemma~1.23]{schwichtenberg2006constructive}
\[
\textsf{mSplit}:
\all{x, y, \epsilon : \dR}  0 < \epsilon \to  \mval \big(\coprodp{(x < y + \epsilon)}{(y < x + \epsilon)}\big)
\]
which is derivable in the theory.
\end{example}
\begin{proof}
Classical trichotomoy (\ref{a:R21}) lets us construct the term
\[
\all{x, y, \epsilon : \dR}  0 < \epsilon \to x < y + \epsilon \lor y < x + \epsilon.
\]
Since the inequalities are semi-decidable we can apply $\textsf{select}$ (\ref{a:M10}) on this term to derive $\textsf{mSplit}$.
\end{proof}

\subsection{Deterministic limits}\label{ss:deterministic limits}
The first case is exactly the ordinary metric completeness which is realized by the primitive limit operations in exact real number computation software.
We introduce it as an axiom in our theory.
For a sequence $f: \dN \to \dR$ we define the property of $f$ being a fast Cauchy sequence by
\[
\mathsf{is\_Cauchy}\,f :\equiv
\all{n,m : \dN}\allx (x : f\,n).\;\allx(y : f\,m).\; -2^{-n-m}\leq x - y \leq 2^{-n-m},
\]
and the property of $x$ being a limit of the sequence $f$ by 
\[
\mathsf{is\_limit}\,x\,f :\equiv
\all{n : \dN} \allx(y : f\,n).\; -2^{-n} \leq x - y \leq 2^{-n}.
\]
The constructive completeness axiom then says that for any fast Cauchy sequence, we can construct its limit point:
\[
  \all{f: \dN \to \dR} \mathsf{is\_Cauchy}\,f \to \some{x : \dR} \mathsf{is\_limit}\,x\,f \label{a:R26}  \tag{Axiom R27}
\]

Often, we already have some classical description $P$ of a real number $x$ and want to show that we can constructively get said number.
The following variation of the limit can thus be useful.
\begin{lemma}
For any predicate $P : \dR \to \Prop$ we can construct the term
\[
     \usome{z : \dR} P\ z \! \to \! \big(\all{ n :\dN} \some{e :\dR} \!\csome{a : \dR} P\ a \land \abs{e - a} \leq 2^{-n}\big)  \!\to\!
     \some{a : \dR} P\ a. 
\]
\end{lemma}
\begin{proof}
The sequence defined in the premise of the statement can be shown to be a fast Cauchy sequence and thus by \ref{a:R26} we can construct its limit point $a$.
Also, using the Archimedan Axiom, we can conclude that the limit point is unique.
Thus, $P\ a$ holds.
\end{proof}
The absolute value function $\lvert x \rvert$
in the above lemma is defined by 
taking the maximum of $x$ and $-x$.
The maximization function is constructed using the limit operation. 
The construction is illustrated in detail later in Section~\ref{s:example-realmax}.
However, when we use the absolute values in the form $|x | \leq y$ for some $x, y : \dR$, 
since the type is in $\Prop$, using classical propositional extensionality we can prove that it is identical to $- y\leq x\leq y$.


\begin{remark}
Classical completeness states that any sequence that is classically Cauchy (i.e. without the additional assumption of a guaranteed speed of convergence), classically admits a limit point, or equivalently any classically non-empty and bounded subset of real numbers classically admits a least upper bound.
The classical completeness can be derived from constructive completeness using our set of axioms.
\end{remark}




\subsection{Deterministic limits of nondeterministic sequences}
Now consider the case where there is a single real number we want to obtain and we have a nondeterministic procedure approximating said number.

Suppose we have a nondeterministic sequence $f : \dN \to \mval\dR$.
We can extend the notion of being a Cauchy sequence and a limit point to nondeterministic sequences by
\[
\mathsf{is\_Cauchy}^\mval\,f :\equiv
\all{n,m : \dN}\allx^\mval (x : f\,n).\;\allx^\mval(y : f\,m).\; |x - y| \leq 2^{-n-m},
\]
and
\[
\mathsf{is\_limit}^\mval\,x\,f :\equiv
\all{n : \dN} \allx^\mval(y : f\,n).\; |x - y| \leq 2^{-n}.
\]
We can prove that for any nondeterministic Cauchy sequence, there 
deterministically and constructively exists the limit.
\begin{theorem}
\label{l:lim}
Within our type theory, we can construct a term of the type
\[
\all{f : \dN \to \mval\dR}
\mathsf{is\_Cauchy}^\mval\;f \to \some{x : \dR}\mathsf{is\_limit}^\mval\;x\;f.
\]
In words, a nondeterministic sequence converges to a point if all possible candidates of the nondeterministic sequence converge to the point.
\end{theorem}
\begin{proof}
Given $f : \dN \to \mval\dR$ with $\mathsf{is\_Cauchy}^\mval\;f$,
we first prove that the limit point is unique, i.e., for any $x, y : \dR$
\[
    \mathsf{is\_limit}^\mval\;x\;f \to \mathsf{is\_limit}^\mval\;y\;f \to x = y.
\]
Thus, by \ref{a:M18} it suffices to show that 
\[
 \mval \some{x : \dR}\mathsf{is\_limit}^\mval\;x\;f.
\]
We can nondeterministically choose a sequence through $f$, i.e. we can show
\[
    \mval \some{s : \dN \to \dR} \all{n : \dN} s\ n \in^\mval f\ n.
\]
By \ref{a:M4} it suffices to show
\[
     \some{s : \dN \to \dR} \all{n : \dN} s\ n \in^\mval f\ n \to \some{x : \dR}\mathsf{is\_limit}^\mval\;x\;f.
\]
As the sequence $s$ is a fast Cauchy sequence it has a limit point $x$ by \ref{a:R26} and we can show $\mathsf{is\_limit}^\mval\;x\;f$.

\end{proof}

As for the deterministic case we often already have a classical description of real numbers 
that we want to construct. For example, 
when we compute a square root of a real number $x : \dR$, we first define it classically 
by $S : \dR \to \Prop :\equiv \lam{y : \dR} x = y \times y$
then prove $\some{y : \dR}S\,y$.

For any real number $x : \dR$ and a classical description of real numbers $S : \dR \to \Prop$, define the notation:
$
x \sim_n S :\equiv \csome{y : \dR}(S\,y) \times |x - y| \leq 2^{-n}$ saying that $x$ approximates a real number represented by $S$ by $2^{-n}$.
Then, we can derive the following version of metric completeness.
\begin{corollary}\label{cor:ndd-limit}
We can derive the term
\[
\usome{x: \dR}S\,x \to
\big(\all{n : \dN}\mval\some{y : \dR} y \sim_n S\big) \to
\some{y : \dR}S\,y.
\]
\end{corollary}

\subsection{Nondeterministic limits}\label{ss:nondet limits}
Suppose we are given a classical description of real numbers $S : \dR \to \Prop$ that is classically sequentially closed, i.e., 
we define $\mathsf{is\_seq\_closed}\,S$ to be the following type:
\[
\all{f : \dN \to \dR}(\all{n : \dN}(f\,n) \sim_n S) \to \csome{x : \dR}(S\,x) \times \mathsf{is\_limit}\,x\,f.
\]
A \emph{nondeterministic refinement procedure} is a procedure that for each natural number $n$ and real number $x_n$ with a promise $x_n \sim_n S$ nondeterministically computes a $2^{-n-1}$ approximation to some real number in $S$ which is at most $2^{-n-1}$ apart from $x_n$. 
Note that $x_n$ and $x_{n+1}$ do not necessarily approximate the same number in $S$.

Formally, we define a nondeterministic refinement procedure as a function $f$ of type
\[
f : \all{n : \dN} \all{x : \dR} x \sim_{n} S \to \mval\some{y : \dR} \big(|x - y| \leq 2^{-n-1}\big) \times \big( y \sim_{n+1} S\big).
\]
We show that given such a nondeterministic refinement procedure, we can apply the $\mval$-dependent choice to nondeterministically get a point in $S$.
We call such a point the limit point of the procedure.
\begin{theorem}\label{thm:nondeterministic-limit}
We can construct a term of type
\begin{align*}
&\all{S : \dR \to \Prop} \mathsf{is\_seq\_closed}\,S \to \\
&\quad\mval\some{y : \dR}y \sim_0 S \to \\
&\quad\big(\all{n : \dN} \all{x : \dR} x \sim_{n} S \to \mval\some{y : \dR} \big(|x - y| \leq 2^{-n-1}\big) \times \big( y \sim_{n+1}S\big)\big)\to\\
&\quad\mval\some{y : \dR}S\,y
\end{align*}
\end{theorem}
\begin{proof}
Let
\[
P\,n :\equiv \some{x : \dR} x \sim_n S\quad\text{and}\quad
R\,n\,x\,y :\equiv |\pi_1\,x - \pi_1\,y| \leq 2^{-n-1}.
\]
The refinement procedure $f$ can be easily adjusted to become an $\mval$-trace 
of $R$.
The $\mval$-dependent choice on $f$ with an initial approximation $x_0 : \mval\some{y : \dR}y \sim_n S$,
yields a sequence $g : \dN \to \dR$ that is consecutively close, i.e.,
\[ \all{n:\dN} |(g\,n) - (g\,(n+1))| \leq 2^{-n-1}, \] 
and converges to an element in $S$, i.e., $\all{n : \dN}g\,n \sim_n S$.
As $S$ is sequentially closed and we can prove that $g$ is Cauchy, applying the ordinary limit on
$S$ constructively yields a point in $S$. 
Hence, applying $\lift^\mval$ on the procedure and postcomposing it to the result of 
the $\mval$-dependent choice yields the nondeterministic limit.
\end{proof}
The nondeterministic refinement procedure defined above is quite generic, and it is often difficult to define such a procedure.
Specifically, it requires to consecutively refine any possible previous approximation to a better approximation of a limit.
In practice, the previous approximation will not be arbitrary, but already derived according to the rules of the procedure.  
It is therefore useful to think of cases where all possible approximations throughout the indefinite refinement procedure share some invariant properties.

To make this more clear, let us again consider the nondeterministic function 
\[ f\,n\,x :\equiv \begin{cases}
0\text{ or }1 &\text{if } n = 0,\\ x&\text{otherwise,}\end{cases}\]
similar to the one from Section~\ref{s:nondet-dep-choice}.

Starting with $1/2$, the function nondeterministically generates the two sequences  $1/2,0,0,\cdots$ and $1/2,1,1\cdots$.
Both are Cauchy sequences that converge to $0$ and $1$ respectively, thus we would consider $0$ and $1$ possible limit points.
However, note that $f$ is not an admissible refinement procedure in the previous sense:
When $2^{-n}$ is given as a $2^{-n}$ approximation to $0$, $f$ returns $2^{-n}$ which is not a $2^{-n-1}$ approximation to any of $0$ or $1$. 
In other words, we lose the information that when applying $f$,
we only encounter either $0$ or $1$ when $n > 0$.
Thus, to apply Theorem~\ref{thm:nondeterministic-limit} we would have to artificially modify the procedure to deal with arbitrary approximations when $n > 0$.
Instead, we would like to use the invariant property of $f$ directly to build a more effective nondeterministic limit operation.

Let $S : \dR \to \Prop$ be a classical description of real numbers that is sequentially closed.
We declare an invariant property of approximations $Q : \dN \to \dR \to \Type$ that is preserved throughout the refinements.
We can encode $Q$ in $P$ at the step of applying the $\mval$-dependent choice:
\[
P\,n :\equiv \some{x : \dR} (x \sim_n S) \times Q\,n\,x.
\]
A similar derivation as in the previous section yields the following more informative limit operation:
\begin{theorem}\label{thm: mlimit advice}
Within our type theory, we can construct a term of type
\begin{align*}
&\all{S : \dR \to \Prop}\all{Q : \dN \to \dR \to \Type} \mathsf{is\_seq\_closed}\,S \to \\
&\quad\mval\some{y : \dR} (y \sim_0 S) \times Q\,0\,y \to \\
&\quad\big(\all{n : \dN} \all{x : \dR} (x \sim_{n} S) \times Q\,n\,x \to\\ 
&\qquad\mval\some{y : \dR} (|x - y| \leq 2^{-n-1}) \times (y \sim_{n+1}S) \times (Q\,(n+1)\,y)\big) \to\\
&\quad\mval\some{y : \dR}S\,y
\end{align*}
\end{theorem}

Note that the required nondeterministic refinement procedure 
accepts additional information on its input $Q\,n\,x : \Type$, on which we can do effective reasoning as it is indexed through $\Type$. For example, in the above case of $0$ and $1$,
we can let $Q \, n\, x :\equiv n> 0 \to (x = 0) + (x = 1)$ 
such that in the beginning of each refinement step $n > 0$,
we can effectively test if $x$ is $0$ or $1$. 
The price to pay is that 
in each step we have to construct a Boolean term
which indicates whether the refinement is $0$ or $1$ which then is used in the next refinement step.
Figure~\ref{fig:Mtrace} illustrates this example.
\begin{figure}
    \centering
    \subfloat[An $\mval$-\emph{trace} of $R\,n\,x\,y \equiv |x-y| \leq 2^{-n}$ converging to \(S=\{0,1\}\)]{{
    \includegraphics[width=0.47\textwidth]{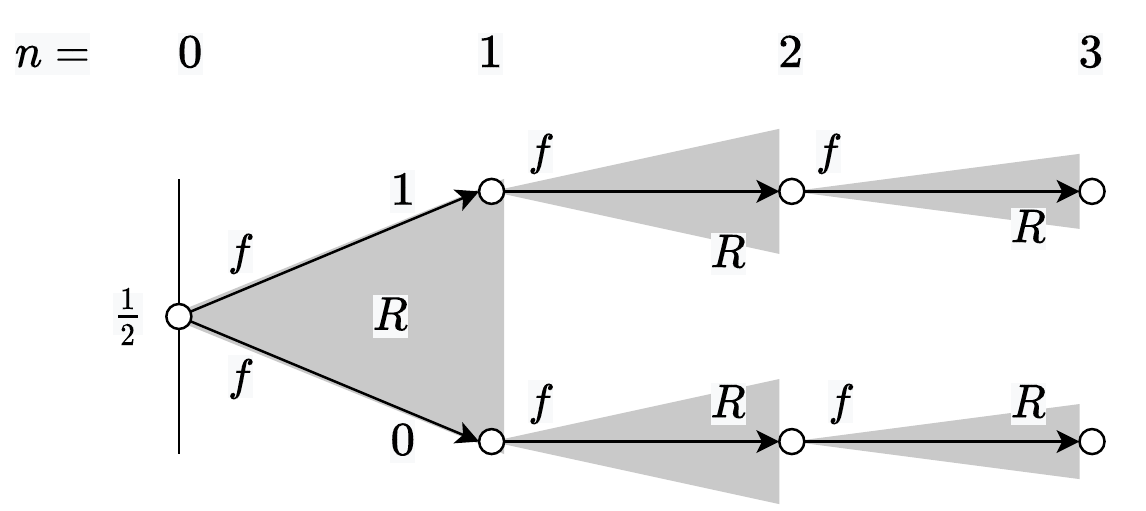}
    }}
    \hfill
    \subfloat[\centering Without the invariant $Q$, the refinement procedure $f$ is not admissible.]{{
    \includegraphics[width=0.47\textwidth]{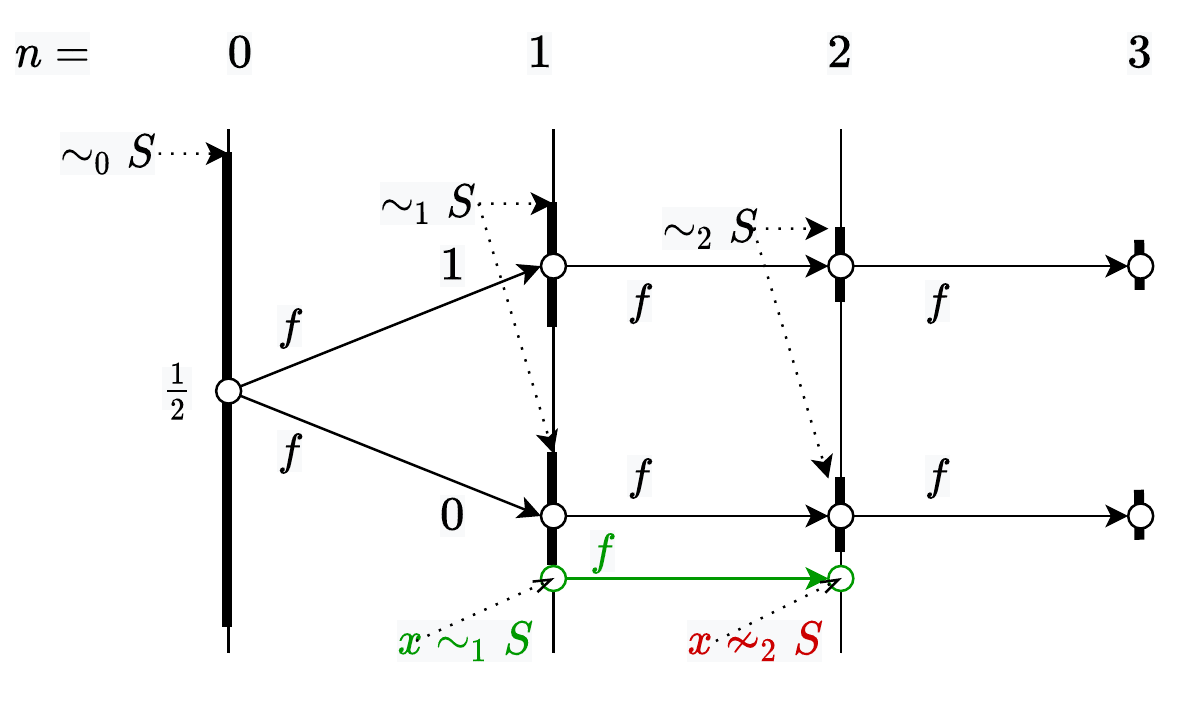}
    }}
    \caption{Using an invariant property of $f$ to define a limit}
    \label{fig:Mtrace}
\end{figure}
\begin{remark}
The iRRAM C++ framework provides a similar operation \code{limit\_mv}. 
The operator computes the limit of a nondeterministic sequence using an additional discrete hint \code{choice} that restricts the possible values of the limit \cite{muller1998implementing}.
\end{remark}
\subsection{Complex numbers and general Euclidean spaces}
Having defined real numbers, it is natural to think of possible extensions such as complex numbers or real vector spaces.
Adding complex numbers to our theory is simple.
We define a complex number as a pair of real numbers, i.e.\ $\dC :\equiv \dR \times \dR$ and define the field operations in the standard way.

We further define a maximum norm $\lvert \cdot \rvert : \dC \to \dR$ and a distance $d\ z_1 \ z_2 := \lvert z_1 - z_2 \rvert$.
For any complex number, we can separate the real and imaginary parts and treat them like ordinary real numbers to prove some of their properties.
Specifically, we can extend all the limit operators defined in this section to complex numbers.
The use of the maximum norm for the distance makes the proofs straightforward.

More generally, for any $n \colon \dN$ we can define the type $\textsf{euclidean}_n$ simply as an $n$-element list of $\dR$ with a maximum norm, and define the vector space operations point-wise.

\section{Soundness by Realizability}
\label{ss:axiomatization:realizability}

To prove soundness of the set of axioms that we present,
we extend the standard realizability interpretation 
that is introduced in Section~\ref{ss:background:typetheory}.
That means we need to declare the interpretations 
of our axiomatic types which are assemblies
and the interpretations of axiomatic terms which are computable functions.
Then, the interpretations of complex types and terms are defined automatically by the original inductive construction of the interpretation.
For example, if we define the interpretation of $\dR$, 
the interpretation of $\dR \to \dR$ is defined 
automatically as the exponent of the interpretation 
of $\dR$.


We interpret $\dK$, of course, as the Kleenean assembly $\bK$ (\ref{a:K1}).
The axiomatic term constant $\ctrue$ is interpreted as the computable point $\ctt : \mathbf{1} \to \bK$
(\ref{a:K2})
and $\cfalse$ is interpreted as the computable point $\cff : \mathbf{1} \to \bK$ (\ref{a:K3}).
See that the interpretation of the type $\ctrue \neq \cfalse$ is $\mathbf{1}$ which is trivially 
inhabited. Hence, we confirm the validity of \ref{a:K4}.

Similarly, we interpret the algebraic operations in \ref{a:K5}-\ref{a:K7}
to the corresponding computable operations of Kleeneans $\bK$ in Figure~\ref{t:erc:kleene}.
See that the related classical properties \ref{a:K8}-\ref{a:K10} get interpreted to the final assembly $\mathbf{1}$.
For the last axiom \ref{a:K11}, see that the interpretation of the type 
$\all{x : \dK} \defined{x} \to \upc{x} + \downc{x}$
is the assembly of functions that decides if $x = \ctt$ or $x = \cff$ given that $x \neq \bot$.
Such a function is realizable by a machine that simply iterates through any given realizer of $x$ checking 
if it encounters $1$ or $2$,
as it is ensured that the realizer is not $0^\IN$ which is the only realizer of $\bot$.

We interpret $\dR$ as any effective assembly of reals $\bR$ (\ref{a:R1}).
We map $0, 1 : \dR$ to the computable points $0, 1 : \mathbf{1} \to \dR$ (\ref{a:R2} - \ref{a:R3})
and map $+, \times, -$ to the corresponding computable functions in $\Casm$ seen in Fact~\ref{f:erc}.
One subtle part is the multiplicative inversion (\ref{a:R8}) which is axiomatized in our type theory as 
\[
/ : \all{x : \dR} x \neq 0 \to \dR .
\]
See that the interpretation of the type is
an assembly of the dependent functions:
\[
f : 
\dR \ni x \mapsto 
\begin{cases}
|\bR^{\mathbf{0}}| &\text{if } x = 0, \\
|\bR^{\mathbf{1}}| &\text{otherwise }, \\
\end{cases}
\]
which gets defined by 
nominating a partial function
$\IR \setminus \{0\} \to \IR$.
Due to the triviality of a function from the empty-set, 
a realizer of such $f$ is allowed to diverge when it 
receives a realizer of $0 \in \IR$.
In our case, we of course interpret 
the division to the multiplicative inversion
which is not defined at $0\in \IR$.

We interpret $< : \dR \to \dR \to \Prop$ as the following function 
\[
x\,\sem{<}\, y := \begin{cases}
\mathbf{1}& \text{if } x < y,  \\
\mathbf{0} &\text{otherwise}
\end{cases}
\]
which is trivially computable (\ref{a:R9}).
Under this interpretation,
we can validate \ref{a:R10}
which is direct from 
that the order relation of real numbers is semi-decidable.

All the classical properties
\ref{a:R10}-\ref{a:R27} are validated that their interpretations get to be $\mathbf{1}$ which requires some clerical works to verify.
The validity of the deterministic limit axiom \ref{a:R26}
is also direct from the computability of limits
of fast Cauchy sequences that is seen from Section~\ref{ss:computable analysis}.

In order to interpret the nondeterminism monad $\mval$, 
we specify an endofunctor $\semmval : \Casm \to \Casm$ in our model.
For an assembly $\mathbf{X}$, $\semmval(\mathbf{X})$ is an assembly of the set of nonempty subsets of $\mathbf{X}$
whose realization relation $\Vdash_{\semmval(\mathbf{X})}$ is defined by
\[
\varphi \Vdash_{\semmval(\mathbf{X})} S \quad :\Leftrightarrow \quad \exists x.\  x\in S \land \varphi \Vdash_\mathbf{X} x\ .
\]
In words, $\varphi$ realizes a nonempty subset $S$ of $\mathbf{X}$ if 
$\varphi$ realizes an element $x$ of $S$ in the original $\bA$.
Observe that for any assemblies $\mathbf{X}, \mathbf{Y}$, a multivalued function $f : \mathbf{X} \mto \mathbf{Y}$ is computable in the sense of computable analysis \cite{w00} if and only if it appears as a morphism $f : \mathbf{X} \to \semmval(\mathbf{Y})$. 

The endofunctor $\semmval$ can be easily proven to be a monad with the unit $\eta_\mathbf{X} : x \mapsto \{x\}$,
the multiplication $\mu_\mathbf{X} : S \mapsto \bigcup_{T \in S} T$, and the action on morphisms $\semmval(f) : S \mapsto \bigcup_{x \in S}\{f(x)\}$.
The computability of the functions hold trivially.

We interpret our nondeterminism 
type former $\mval : \Type \to \Type$
as the monad (\ref{a:M1}) $\semmval : \Casm \to \Casm$
and interpret the unit (\ref{a:M2}), 
the multiplication (\ref{a:M3}), and 
the function lift (\ref{a:M4})
appropriately as the computable functions
of the monad structure.
Then, the types of the coherence conditions, 
\ref{a:M5}-\ref{a:M9}, get interpreted to $\mathbf{1}$.

To validate the core characterization of the nondeterminism which is based on the relation with the classical non-empty power-set monad $\cval$, let us see the interpretation of the monad in the model.
See that for any type $X$, it holds that
\[
|\sem{\cval(X)}| = \{
(S: \sem{X} \to \{\mathbf{1}, \mathbf{0}\}, \star) \mid
\exists x.\; S\;x = \mathbf{1}
\}
\]
where any $\varphi$ realizes any $(S, \star)$
as any $S :\sem{X} \to \sem{\Prop}$ is trivially 
computable.

Calling this induced monad $\mathbf{P}_+ : \Casm \to \Casm$,
see that the underlying set 
of $\mathbf{P}_+(\mathbf{X})$ is set-theoretically isomorphic 
to the underlying set of $\mathbf{M}(\mathbf{X})$ for any 
assembly $\mathbf{X}$. And, for any assembly 
$\bA$, the assembly $\mathbf{P}_+(\mathbf{X})$ has a trivial realization relation.
Hence, there is a monic (and set-theoretically surjective) natural transformation defined for each $\bA$:
\[
\mathbf{picture}_\mathbf{X} : \mathbf{M}(\mathbf{X}) \ni S \mapsto (S, \star) \in
\mathbf{P}_+(\mathbf{X})
\]
which can be trivially computed 
as $\mathbf{P}_+(\mathbf{X})$ has the trivial realization relation.
We interpret the natural transformation $\pic$ (\ref{a:M10})
to the natural transformation $\mathbf{picture}$.
Other classical properties 
\ref{a:M11} - \ref{a:M14} can be validated.

Since the underlying sets of the interpretations 
$\cval(\mval X)$ and $\cval (\cval X)$ are
set-theoretically isomorphic and they both have
trivial realization relation, it can be easily 
proved that
\[
\mathbf{P}_+(\mathbf{picture}_{\mathbf{X}})
: 
\mathbf{P}_+(\semmval(\mathbf{X})) \to
\mathbf{P}_+(\mathbf{P}_+(\mathbf{X}))
\]
is an isormorphism in $\Casm$, for any $\mathbf{X}$,
whose inverse is also trivially computable.
This observation validates \ref{a:M15}.
The destruction method \ref{a:M16} becomes obvious under this interpretation.

When we have two Kleeneans $x, y \in \bK$, 
excluding the domain where both are non-$\ctt$, 
we can devise a machine which iterates through the given realizers of $x$ and $y$ to check if $1$ appears in the given realizer of $x$ or in the given realizer of $y$.
Calling this machine $\varphi$, for $\varphi_x \Vdash_\bK x$
and $\varphi_y \Vdash_\bK y$, 
$\varphi(\varphi_x, \varphi_y)$ realizes 
$\{\ctt, \cff\}$ when $x = y = \ctt$,
$\{\ctt\}$ when $x = \ctt\land y \neq \ctt$,
and $\{\cff\}$ when $x \neq \ctt \land y = \ctt$
in $\semmval(\bK)$.
See that this procedure, $\semselect$, which is obviously computable, realizes the interpretation of \ref{a:M17}.
When $\mathbf{A}$ is sub-singleton, $\semmval(\mathbf{A})$ is actually identical to $\mathbf{A}$. Hence, \ref{a:M18} gets validated.
As it was discussed in the section, the last axiom \ref{a:M19} gets validated by the fact that we can form recursion on the level of realizers.

Discussions thus far conclude the soundness of our axioms:
\begin{metatheorem}
The axiomatization is sound admitting a realizability interpretation.
\end{metatheorem}

\section{Relating Classical Analysis}
\label{s:relator}
As described in the previous sections, our theory allows classical reasoning for non-computational properties. 
For example, to prove a statement of the form 
\[ \all{x : \dR}\some{y : \dR}P\ x\ y \]
where $P : \dR \to \dR \to \Prop$,  we assume any $x : \dR$, provide an explicit $y : \dR$ and prove that $P\ x\ y$ holds. 
$P\ x\ y : \Prop$ is a classical statement and thus we do not need to restrict ourselves to constructive reasoning in its proof.
In terms of program extraction (cf. Section~\ref{s:implementation}), the computational content of the statement is just a program that maps any $x : \dR$ to a $y : \dR$ and the information that $P\ x\ y$ holds is not part of the program itself.

In a sense we can thus often separate a proof in a constructive part, which we would like to map to a program in a certain way, and a purely mathematical part, where we are mostly concerned about correctness and care less about the exact form of the proof.
Now, this poses the practical question on how much of classical mathematics needs to be re-proven over our real number type $\dR$.

Of course, there are already several formalizations of the classical real numbers in proof assistants, which already provide many of the statements we need. 
Instead of starting from scratch, we would like to access such results and be able to apply them in our proofs.

To make this precise, let us assume we have a second real number type $\dRc$ and we have already proven some results over $\dRc$.
Let us further assume that different from our formalization, $\dRc$ is completely classical, by which we mean that classical statements such as  
\[
\all{x : \dRc} \coprodc{x > 0}{\neg (x > 0)}
\]
hold.
Readers familiar with the Coq proof assistant can for example think of the real number type in the Coq standard library.
We want to embed our axiomatization and apply theorems proven over the classical theory to our formalization while separating the constructive part and the classical part of the type theory correctly so that realizability results like those from Section~\ref{ss:axiomatization:realizability} still hold.

Note that, even though now the type theory provides classical types and terms through $\dRc$, 
it should stay fully constructive for the terms that do not access any of the classical axioms.
We can thus think of a term in the type theory as to be formally interpreted in two different models.

Let us define the two type judgements 
\[ \classicaltyping   t : X \]
saying that $t$ of type $X$ may rely on classical axioms and 
\[ \constructivetyping   t' : X' \]
saying that $t'$ of type $X'$ is defined without any classical axioms.

When
$\classicaltyping   t : X$, we interpret it in the category of sets $\Cset$ and
when $\constructivetyping   t : X$, we interpret it in $\Casm$. 
\begin{example}
There is a term $t$ such that
\[ \classicaltyping  t: \all{x : \dRc} \coprodc{x > 0}{\neg (x > 0)} \]
is derivable but 
\[ \constructivetyping  t : \all{x : \dRc} \coprodc{x > 0}{\neg (x > 0)} \] 
does not hold.
\end{example}

Our goal is to relate the two type judgements. 
As we do not assume any counter-classical axioms such as the continuity principle in our theory, 
every constructively proven statement is also true classically.
Thus, obviously when $\constructivetyping   t : X$ is derivable, so is $\classicaltyping   t : X$.

However, we are interested in the other direction, i.e. how to get a constructively well-typed term from a classically well-typed term.
$\Cset$ is a reflective subcategory of $\Casm$
by the forgetful functor $\semGamma : \Casm \to \Cset$ and its right adjoint 
$\semnabla : \Cset \to \Casm$ where for any set $X$, $\semnabla X$ is the assembly of $X$ with 
the trivial realization relation \cite[Theorem~1.5.2]{van2008realizability}. 

For each type $X$, define 
\[
\nabla X :\equiv \some{P : X \to \Prop}\usome{ x : X}P\ x\;.
\]
See that for any type $X$, $\sem{\classicaltyping \nabla X : \Type}$ is isomorphic to $\sem{\classicaltyping X  : \Type}$ in $\Cset$ and 
$\sem{\constructivetyping \nabla X : \Type}$ is isomorphic to $\semnabla\semGamma\sem{\constructivetyping X : \Type}$ in $\Casm$. 
It can be understood as a functor that erases all the computational structure of $X$
while keeping its set-theoretic structure.

It is provable in the type theory using the assumptions of $\Prop$ being the type of classical propositions admitting propositional extensionality 
that $\nabla$ is an idempotent monad where its unit 
$\nablaunit : \all{X : \Type} X \to \nabla X$ on $\nabla X$, 
is an equivalence with the inverse being the multiplication. 
Moreover, it holds that $\nablaunit\ \Prop : \Prop \to \nabla \Prop$ is an equivalence.
That means, given a mapping $f : X_1 \to X_2 \to \cdots \to X_d$, 
there is a naturally defined lifting $\nablalift{f} : \nabla X_1 \to \nabla X_2 \to \cdots \to \nabla X_d$
and given a predicate $P : X_1 \to X_2 \to \cdots \to \Prop$, 
there is $\nablalift{P} :\nabla X_1 \to \nabla X_2 \to \cdots \to \Prop$.

We add the type judgement rule:
\[
\infer[X \text{ is transferable } (\relaterule)]
{\constructivetyping   t :  X}
{\classicaltyping  t :  X}
\]
saying that when $t$ is a classically constructed term of a transferable type $X$, 
we have a constructive term $t$ of type $X$.
A type is transferable if it is of the form $\nabla c$ for a constant type, a $\Type$, or a $\Prop$ variable $c$;
$X \to Y$, $X \times Y$, $X \land Y$, $X \lor Y$, or $\nabla (X + Y)$ for transferable types $X$ and $Y$; $\all{x : X}P(x)$, $\some{x : X}P(x)$, or $\csome{x : X}P(x)$ for transferable types $X$ and $P(x)$; $x = y$, $x < y$, or $x <^\dagger y$; or is $\Type$ or $\Prop$. Roughly speaking, a type is transferable if its subexpressions in the construction of the type are guarded by $\nabla$.

This judgement rule is validated in our interpretation. First, note that $\semnabla\{*\} \simeq \terminalobj$. And, when a type $X$ is transferable, $\sem{\constructivetyping X : \Type} \simeq \semnabla\sem{\classicaltyping X : \Type}$ holds. When $\classicaltyping   t : X$, we have a function $\sem{\classicaltyping   t : X} : \{*\} \to \sem{\classicaltyping   X : \Type}$
in $\Cset$. Hence, we define $\sem{\constructivetyping  t :  X} : \terminalobj \to \sem{\constructivetyping X : \Type}$ by pre and postcomposing the above isomorphisms to
 $\semnabla\sem{\classicaltyping   t : X} : \semnabla\{*\} \to \semnabla\sem{\classicaltyping  X : \Type}$.

To relate our axiomatic real numbers with the classical real number type $\dRc$, we assume the existence of a map $\relator$ with interpretation in $\Cset$ being the identity map $\sem{\classicaltyping \relator :\dR \to \nabla \dRc} :\IR \ni x  \mapsto x \in \IR$.
We assume the following axioms that characterize the mapping.
\begin{flalign}
& \label{a:N1} \relator : \dR \to \nabla \dRc \tag{Axiom $\nabla$1} \\
& \label{a:N2} \all{ x,  y :\dR} \relator\ x = \relator\ y \to x = y \tag{Axiom $\nabla$2} \\
& \label{a:N3} \all{y : \nabla \dRc}\csome{x :\dR} y = \relator\ x \tag{Axiom $\nabla$3} \\
& \label{a:N4} \relator\ 0 = \nablaunit\ \dRc\ 0 \tag{Axiom $\nabla$4} \\
& \label{a:N5} \relator\ 1 = \nablaunit\ \dRc\ 1 \tag{Axiom $\nabla$5} \\
& \label{a:N6} \all{x, y : \dR}\relator\ (x + y) = (\relator\ x) \nablalift{+} (\relator\ y) \tag{Axiom $\nabla$6} \\
& \label{a:N7} \all{x, y : \dR}\relator\ (x \times y) = (\relator\ x) \nablalift{\times} (\relator\ y) \tag{Axiom $\nabla$7} \\
& \label{a:N8} \all{x : \dR}\relator\ (-x) = \nablalift{-} (\relator\ x) \tag{Axiom $\nabla$8} \\
& \label{a:N9} \all{x : \dR}\all{p : x \neq 0}\relator\ (/x\ p) = \nablalift{/} (\relator\ x) \tag{Axiom $\nabla$9} \\
& \label{a:N10} \all{ x, y : \dR} (x < y) = (\relator\ x) \nablalift{<} (\relator\ y) \tag{Axiom $\nabla$10} 
\end{flalign}

\begin{example}
Using $\relator$, we can prove the classical properties of $\dR$
without using the classical axioms of $\dR$. 
For example, suppose we want to prove 
\[ \all{x, y : \dR} x + y = y + x. \]
Using the classical real numbers, we can trivially prove
\[
\classicaltyping\textsf{comm}_{\dRc} :\all{x, y : \nabla \dRc} x \nablalift{+} y = y \nablalift{+} x\;.
\]
Since $x,y :\dRc \classicaltyping   x \nablalift{+} y : \nabla \dRc$, 
from the added judgement rule, we have
\[
x,y :\dRc \constructivetyping   x \nablalift{+} y : \nabla \dRc.
\]
Therefore, the above proof $\textsf{comm}_{\dRc}$ directly transfers to 
a constructive proof 
\[ \constructivetyping  \textsf{comm}_{\dRc} :\all{x, y : \nabla \dRc} x \nablalift{+} y = y \nablalift{+} x. \]

Similarly, from the order decidability of the classical real,
\[
\all{x : \dRc} (x > 0) + (x = 0) + (x < 0)
\]
we can lift it to obtain
\[
\all{x : \nabla \dRc} (x > 0) \lor (x = 0) \lor (x < 0)
\]
Now, using the relator, we can obtain the classical total order
\[
\all{x : \dR} (x > 0) \lor (x = 0) \lor (x < 0)
\]
\end{example}
We will look at a more practical example in Section~\ref{s:square root}.

Note that in Coq we can not formally deal with having two independent type theories simultaneously and therefore a complete correctness proof when applying the relator notion can only be formulated on the meta-level.
We plan to address this issue in future work e.g.\ by writing a Coq plugin that would allow this distinction.

\section{Examples}\label{s:examples}
Let us illustrate the use of our formalization by applications to some standard examples from computable analysis.
\subsection{Maximization}\label{s:example-realmax}
A simple example of an operation that requires multivaluedness in its definition is the maximization operator that takes two real numbers $x$ and $y$ and returns their maximum.
That is, we define 
\[
\mathsf{max}\ x \ y \ z :\equiv  (x > y \rightarrow z = x) \land (x = y \rightarrow z = x) \land (x < y \rightarrow z = y)
\]
and want to derive the term
\[
   \all{x\ y : \dR} \some{z : \dR} \mathsf{max}\ x\ y\ z. 
\]

Partiality of comparisons prevents a direct proof of the statement.
Instead, we apply the limit operator defined in Corollary~\ref{cor:ndd-limit}.
That is, we need to show that there is exactly one $z : \dR$ for which $\mathsf{max}\ x \ y\ z$ holds and that
\[
    \all{n : \dN} \mval \some{z' : \dR} z' \sim_n z. 
\]
The first statement follows from the classical trichotomy.
For the second one we use multivalued branching (Example~\ref{exa:split}) to construct the nondeterministic sequence
\[
    m_{x,y}\ n :\equiv \begin{cases}
    x &\text{if } x > y - 2^{-n}, \\
    y &\text{if } y > x - 2^{-n}.
    \end{cases}
\]
That is, we construct the approximation by concurrently testing whether $x > y - 2^{-n}$, in which case $x$ is a good enough approximation for the maximum, or $x < y + 2^{-n}$, in which case $y$ is a good enough approximation.

\subsection{Intermediate Value Theorem (IVT)}
A classical example from computable analysis (see e.g. \cite[Chapter~6.3]{w00}) is finding the zero of a continuous, real valued function $f: [0,1] \to \IR$ with $f(0) < 0$ and $f(1) > 0$ under the assumption that there is exactly one zero in the interval (i.e. a constructive version of the intermediate value theorem from analysis).

We can define continuity of a function $f: \dR \to \dR$ using the usual  $\epsilon$-$\delta$-criterion
\[
    \mathsf{continuous} \ f :\equiv \all{x : \dR} \all{\epsilon > 0} \csome{\delta > 0} \all{y : \dR} \abs{x - y} \to \abs{(f \ x) - (f\ y)}< \epsilon
\]
and define the property of $f$ having exactly one zero in the interval $[a,b]$ by
\[
    \mathsf{uniq}\ f\ a \ b :\equiv a < b \land f\ a < 0 \land f\ b > 0 \land \usome{z : \dR} a < z < b \land f\ z = 0.
\]
Here $\all{\epsilon > 0} P$ is short for $\all{\epsilon : \dR} \epsilon > 0 \to P$ and $\csome{\epsilon > 0} P$ is short for $\csome{\epsilon : \dR} \epsilon > 0 \land P$.
Then the property we want to show is
\[
\mathsf{IVP}: \all{f : \dR \to \dR} \mathsf{continuous}\ f \to \mathsf{uniq}\ f\ 0\ 1 \to \some{z : \dR} 0<z<1 \land f\ z = 0.
\]

The statement can be proven using the trisection method which is similar to the classical bisection method but avoids uncomputable comparison to $0$.
That is we inductively define sequences $a_i, b_i$ with $f(a_i) \times f(b_i) < 0$ and $b_i - a_i \leq (2/3)^i$.
In each step we let $a_i' := (2a_i + b_i)/3$, $b_i' := (a_i +2b_i)/3$ and in parallel check if $f(a_i') \times f(b_i) < 0$ or $f(a_i) \times  f(b_i') < 0$. 
In the first case we set $a_{i+1} := a_i'$, $b_{i+1} := b_i$, in the second case $a_{i+1} := a_i$, $b_{i+1} := b_i'$.
As at least one of the inequalities is true by the assumptions,  this selection can be done using the $\select$ operator from Section~\ref{s:nondeterminism}. 
The zero can then be defined using the limit operator.
\subsection{Classical proofs and a fast square root algorithm}\label{s:square root}
Let us now look at an example on how the \relator{} operator from Section~\ref{s:relator} can be used.
Again, suppose we introduced a type $\dRc$ of classical real numbers in our theory that already provides proofs of some classical facts about real numbers.
In particular, suppose we have a term $sq$
saying that for any positive real number, there is a square root:
\[
\classicaltyping sq : \all{x :\dRc}0<x \to \some{y : \dRc} x = y \times y
\]
As $\nablaunit\ A : A \to \nabla A$ is an equivalence in the classical type theory, we can derive
\[
\classicaltyping sq' : \all{x :\nabla \dRc}
(\nablaunit\ 0) \nablalift{<} x  \to \some{y : \nabla \dRc} x = y \nablalift{\times} y
\]
As the type of the above judgement is transferable, we have
\[
\constructivetyping sq' : \all{x :\nabla \dRc}
(\nablaunit\ 0) \nablalift{<} x  \to \some{y : \nabla \dRc} x = y \nablalift{\times} y
\]
Using the axioms of the relator, we can obtain a term of type
\[
\constructivetyping \all{x :\dR}
0 < x \to \csome{y : \dR} x = y \times y : \Prop.
\]
This illustrates how we can transport a classical proof of the existence of square root based on $\dRc$ to a constructive proof of the classical existence of square root based on $\dR$.

Let us assume that our classical theory already contains all facts about the square root that we need and that similar to the above example we can transport the classical proofs into our constructive setting.
We use the transported results to verify a constructive version of the real square root function which can be used for efficient computation of the square root, thus we only need to focus on the constructive parts of the below proof.

By the real square root of $x$ we mean the unique nonnegative number $y$ such that $x = y \times y$, and use $y = \sqrt{x}$ as shorthand notation for the property $x = y \times y \land y \geq 0$.
That is, we want to show
\[
    \mathsf{rsqrt}: \all{x : \dR} x \geq 0 \to \some{y : \dR} y = \sqrt{x}
\]
The idea on how to construct this $y$ is fairly simple and has already been used e.g. as an example in \cite{konevcny2020computable}.
The main steps are as follows.
\begin{enumerate}
    \item We define a restricted version of $\mathsf{rsqrt}$ in the interval $[0.25,2]$ using the well-known Heron method.
    \item For any number $x > 0$, we can find a $z \in \IZ$ such that $4^z x \in [0.25, 2]$, i.e., we can scale the number to be inside the interval, apply the restricted square root and rescale the result accordingly.
    \item We extend to $x \geq 0$ by approximating the square root of small enough numbers by $0$.
\end{enumerate}
Let us start with the first item, i.e., we show
\[
    \all{x : \dR} \frac{1}{4} \leq x \leq 2 \to \some{y : \dR} y = \sqrt{x}.
\]
We apply constructive completeness \eqref{a:R26}, thus we need to show the classical fact that there is exactly one such number $\sqrt{x}$ and 
\begin{equation}\label{eq:sqrestrict}
    \all{n : \dN} \some{y_n : \dR} \csome{y : \dR} y = \sqrt{x} \land \abs{y_n - y} < 2^{-n}.
\end{equation}
We derive the former from our classical theory.
For the latter we define the sequence 
\begin{align*}
  \mathsf{heron}\ x\ 0 & :\equiv 1 \\ 
  \mathsf{heron}\ x \ n+1 & :\equiv  \frac{1}{2} \times \left ( \mathsf{heron}\ x \ n + \frac{x}{\mathsf{heron}\ x \ n } \right ).
\end{align*}
It is well-known that the Heron method converges quadratically to $\sqrt{x}$ in the interval $[0.25,2]$, i.e. $\abs{\mathsf{heron}\ x\ n - \sqrt{x}} \leq 2^{-2^{n}}$ which can be used to show \eqref{eq:sqrestrict}.

For the second step, i.e., extending to all positive reals, we first show
\[
    \all{x : \dR} \mval \some{z : \dZ} \frac{1}{4} \leq (4^{z} \times x) \leq 2 
\]
The existence of such a $z$ basically follows from the Archimedean axiom and the Markov principle, but note that we can only find such a $z$ nondeterministically.
As $y = \sqrt{4^z \times x} \to 2^{-z} \times y = \sqrt{x}$ we can show
\[
    \all{x : \dR} x > 0 \to \mval (\some{y : \dR} y = \sqrt{x})
\]
and by the uniqueness of the square root \ref{a:M18} lets us remove the nondeterminism.

Finally, for the last step of extending to $0$, we use the limit operation from Corollary~\ref{cor:ndd-limit}.
That is, we need to show
\[
    \all{n : \dN} \mval \some{y_n : \dR} \csome{y : \dR} y = \sqrt{x} \land \abs{y_n-y} \leq 2^{-n}.
\]
Given $n : \dN$, we can nondeterminstically choose one of $x < 2^{-2n}$ or $x > 0$.
In the first case we can choose $0$ as a $2^{-n}$ approximation to the square root.
In the second case we can get an exact square root from step 2.
\begin{remark}
Step 2 and 3 can also be combined in a single step leading to a slightly more efficient algorithm, but we separated the two parts for the sake of presentation.
\end{remark}
\subsection{Nondeterministic computation of complex square roots}\label{ss:complex-root}
When trying to extend the previous example to complex numbers, a difficulty arises that the square root operation on complex numbers is inherently multivalued.
The square root of a number $z \in \IC$ is a number $x \in \IC$ such that $x^2 = z$.
If $x$ is a square root of $z$ then so is $-x$ and there are no other square roots. 
Thus, for every $z \neq 0$ there are exactly two square roots.
In the case of the square root on nonnegative reals in the previous section we could simply choose one of the two square roots to get a singlevalued branch.
However, it is well known that no such continuous choice exists for the whole complex plane:
The square root has a branch point at $z=0$ and thus there is no singlevalued, continuous square root function in any region containing $z=0$ as an interior point.
We show that we can, however, prove the nondeterministic existence of a square root constructively in our theory.
That is, we prove the constructive and nondeterministic existence of square roots of complex numbers using a simple method described e.g. in \cite{muller1998implementing}.

The following well-known algebraic formula can be used to reduce the calculation of complex square roots to calculating real square roots (see e.g. \cite[\S 6]{cooke2008classical}).

Let $z = a + ib$, then
\[
 \sqrt{\frac{\sqrt{a^2+b^2}+a}{2}}+i\cdot\mathrm{sgn}(b)\cdot\sqrt{\frac{\sqrt{a^2+b^2}-a}{2}}
\]
is one of the square roots of $z$.
Of course, this function is not computable as $\mathrm{sgn}$ is not continuous in $0$.
However, if $z \neq 0$, we can nondeterministically choose one of the cases $a < 0$, $a > 0$, $b < 0$, $b >0$  and apply the formula.
(In case $a > 0$ or $a < 0$, a slight adaption of the formula using $\mathrm{sgn}(a)$ instead of $\mathrm{sgn}(b)$ is used.)

Thus, using the real square root we can show the following restricted version of the existence of a complex square root
\begin{equation}\label{eq: restricted sqrt}
    \sqrt{}_0: \all{z : \dC} z \neq 0 \to \mval (\some{x : \dC} x\times x = z)
\end{equation}
Finally, we apply Theorem~\ref{thm: mlimit advice} to also include the case $z = 0$.
Recall that given a $2^{-n}$ approximation $x_n$ of a square root of $z$ that satisfies a certain predicate $Q$ that we will define later, we need to choose a $2^{-(n+1)}$ approximation $x_{n+1}$ of a square root of $z$ with $\lvert x_{n+1} - x_n \rvert \leq 2^{n+1}$ and such that $x_{n+1}$ satisfies $Q$.
We proceed as follows.
In the beginning, at each step $n$ we nondeterministically choose one of the two cases  $\abs{z} < 2^{-2(n+2)}$ or $\abs{z} > 0$.
In the first case, $0$ is a good enough approximation for any square root of $z$.
In the second case, we know $z \neq 0$ and thus can apply \eqref{eq: restricted sqrt} to get the exact value of a square root.
However, once we have selected the second case, for any later elements of the sequence we just return the previous value $x_n$.
Thus, all possible sequences the refinement procedure returns have the form $0, 0, 0, \dots, 0, x, x, x, \dots$, where $x$ is a square root of $z$.
Further, if we returned $0$ at the $n$-th step, we know that $\abs{z} < 2^{-2(n+2)}$ and therefore for any square root $x$, $\abs{x} < 2^{-(n+2)}$ and returning $x$ at step $n+1$ is a valid refinement of the previous approximation.

Thus, the invariant property of the sequence defined in this way is given by the relation $Q\,n\,x : \Type$ defined by 
\[ 
    Q\ n\ x :\equiv \big((\abs{z} \leq 2^{-2(n+2)}) \times (x = 0)\big) + (x \times x = z).
\]
Applying Theorem~\ref{thm: mlimit advice} with this $Q$, we get
\[
    \sqrt{}: \all{z : \dC} \mval (\some{x : \dC} x\times x = z)
\]
\subsection{Equivalence of axiomatic real numbers}
\label{ss:axiom-equivalence}
To prove that the set of axioms we devised to express exact real number computation is 
expressive enough, we prove that any two types $\dR_1$ and $\dR_2$  satisfying the set of axioms are 
type-theoretically equivalent.
As our type theory is extensional, they are equivalent if we can construct the mutually inverse functions $\iota_1 : \dR_1 \to \dR_2$ and 
$\iota_2 : \dR_2 \to \dR_1$. 
The basic idea of the construction is similar to \cite{DBLP:journals/mlq/Hertling99}
where an effective model-theoretic structure of real numbers is suggested.

From the classical Archimedean principle of real numbers, for any $x : \dR_1$, 
there classically is $z : \dZ$
which bounds the magnitude of $x$ in the sense that
$ |x| < z$ holds.
Applying nondeterministically the Markov principle, we can construct the nondeterministic rounding operator:
\[
\mathsf{round} : \all{x : \dR_1}\mval\some{z : \dZ} 
z - 1 < x < z + 1.
\]
Recall that the usual rounding is not computable due to discontinuity of the classical rounding function \cite[Theorem~4.3.1]{w00}.
Then, by scaling, we can construct a term of type
\[
\mathsf{dyadic} : 
\all{x : \dR_1}\all{n : \dN}\mval\some{z : \dZ}|x - z \times 2^{-n}| \leq 2^{-n}
\]
which nondeterministically approximates the binary magnitude of real numbers.

By using the destruction principle of the nondeterminism
and doing some clerical work, we get 
the fact that for any real number $x : \dR_1$, there exists
a sequence of nondeterministic integers
$f : \dN \to \mval\dZ$ such that 
every section $g :\dN \to \dZ$ of $f$
is an approximation sequence of $x$
in the sense that $\mathsf{is\_limit}\,
x\,(\lam{n : \dN} (g\, n)\times 2^{-n})
$ holds.

Note that the description thus far implicitly used
the integer embedding in $\dR_1$.
Taking out the embedding explicitly, we
can prove that
for any sequence of integers $g : \dN \to \dZ$,
if its induced dyadic sequence is Cauchy in one type of real numbers,
it also is Cauchy in the other one.
Hence, from $f$, using the other integer embedding $\dZ \to \dR_2$,
we can get a sequence of nondeterministic real numbers
in $\dR_2$ where every section is a Cauchy sequence in $\dR_2$.
Thus, after proving that the limit points of such sequences is unique, 
we can apply the deterministic limit of nondeterministic sequences
(Lemma~\ref{l:lim}) to construct a real number in $\dR_2$.

Intuitively, we use the space of sequences of nondeterministic
integers as an independent stepping stone connecting the two 
axiomatic types $\dR_1$ and $\dR_2$. 
In our axiomatization of nondeterminism we can analyze each section of a sequence of nondeterministic integers so that we can apply a limit operation in $\dR_2$.

The other direction $\dR_2 \to \dR_1$ can be constructed analogously, and the two mappings being inverse to each other can be proved easily, concluding that the two axiomatic types $\dR_1$ and $\dR_2$ are equivalent.

\section{Implementation}\label{s:implementation}
We implemented the above theory in the Coq proof assistant\footnote{The source code and detailed usage instructions are on\\ \url{https://github.com/holgerthies/coq-aern/tree/release-2022-01}}.
From a correctness proof in our implementation, we can extract Haskell code that uses the AERN library to perform basic real number arithmetic operations.
For this, we introduce several extraction rules replacing operations on the constructive reals with the corresponding AERN functions.
The extracted code requires only minor mechanical editing, namely adding several import statements (including AERN and comparison operators that return generalised Booleans) and type coercions used in the extracted code to replace erased Coq type translations.


\subsection{Overview of the implementation}

Our axioms for $\kleene$ $\mval$ and $\dR$ are not directly encoded as Coq axioms, but as members of the following type classes:
\begin{itemize}
    \item \mintinline{coq}{LazyBool}: Kleenean operations and their properties, in file \texttt{Kleene.v}
    \item \mintinline{coq}{MultivalueMonad}: operations and properties of $\mval$ including \mintinline{coq}{choose} (a slightly generalized $\select$) and M-dependent choice, in file \texttt{MultivalueMonad.v}
    \item \mintinline{coq}{ComplArchiSemiDecOrderedField}: operations and properties of $\dR$, including limits, in file \texttt{RealAxioms.v}
\end{itemize}
The class \mintinline{coq}{MultivalueMonad} requires also that $\mval$ is an instance of a type class \mintinline{coq}{Monad} and is homomorphic to the set monad \mintinline{coq}{MPset_Monad} via the type class \mintinline{coq}{Monoid_hom}, formally encoding the ``picture'' natural transformation.

We add Coq axioms that declare the existence of models for these type classes:
\begin{minted}{coq}
Parameter K : Set.
Axiom K_LazyBool : LazyBool K.

Parameter M : Type -> Type.
Axiom M_Monad : Monad M.
Axiom MultivalueMonad_description : Monoid_hom M_Monad NPset_Monad.
Axiom M_MultivalueMonad : 
    @MultivalueMonad _ K_LazyBool _ _ MultivalueMonad_description.

Parameter R : Set.
Axiom R_SemiDecOrderedField : @SemiDecOrderedField  _ K_LazyBool R.
Axiom R_ComplArchiSemiDecOrderedField : 
    @ComplArchiSemiDecOrderedField _ _ _ R_SemiDecOrderedField.
\end{minted}
but we do this only locally, at the point of extracting code.  
Our application examples do not depend on these axioms, but instead assume the existence of such models in their context.  For example, our real maximisation implementation includes:

\begin{minted}{coq}
Require Import Real.

Section Minmax.
  Generalizable Variables K M Real.

  Context `{klb : LazyBool K} `{M_Monad : Monad M}
          {MultivalueMonad_description : Monoid_hom M_Monad NPset_Monad} 
          {M_MultivalueMonad : MultivalueMonad}
          {Real : Type}
          {SemiDecOrderedField_Real : SemiDecOrderedField Real}
          {ComplArchiSemiDecOrderedField_Real : 
                                    ComplArchiSemiDecOrderedField}.
          
  Definition real_is_max (x y z : Real)
    := (x > y -> z = x) /\ (x = y -> z = x) /\ (x < y -> z = y).
    
  Lemma real_max_prop : forall x y, {z | real_is_max x y z}.
  Proof....

\end{minted}

\subsection{Exact real computation and the AERN framework}

\begin{figure}[htp]
\begin{minted}[fontsize=\small]{Haskell}
real_max :: CReal -> CReal -> CReal
real_max x y =
  limit $ \(n :: Integer) -> 
    let e = 0.5^n in
    if select (x > y - e) (y > x - e)
      then x
      else y
\end{minted}
    \caption{Real number maximisation in Haskell/AERN}
    \label{fig:aern-realmax}
\end{figure}

Figure~\ref{fig:aern-realmax} lists a hand-written Haskell/AERN code for computing real number maximisation function presented in Subsection~\ref{s:example-realmax}.
This code shows the most important features of AERN.
\mintinline{Haskell}{CReal} is the type of lazy real numbers, essentially a convergent sequence of interval approximations to the real number.  
The sequence is not necessarily fast converging.  
Like in iRRAM, the elements of the sequence are computed with increasing precision for dyadic interval endpoints.  
In practice, the sequences usually converge exponentially (measured in bit-length) since the precision increases exponentially.
The usual Haskell Prelude arithmetic operators are being applied to \mintinline{Haskell}{CReal} numbers in expressions such as  \mintinline{Haskell}{0.5^n} and \mintinline{Haskell}{y - e}.

The real number comparisons \mintinline{Haskell}{x > y - e} and \mintinline{Haskell}{y > x - e} do not use the usual Haskell Prelude comparison operators.  
These comparisons return \mintinline{Haskell}{CKleenean}, a converging sequence of Kleeneans computed with increasing precisions, exactly like elements of \mintinline{Haskell}{CReal}.
The function \mintinline{Haskell}{select :: CKleenean -> CKleenean -> Bool} behaves exactly as the $\select$ operator defined in \ref{a:M17}.  
Note the absence of $\mval$ in the return type of AERN's \mintinline{Haskell}{select}.  
Multivaluedness is intrinsic thanks to redundancy in the underlying representations of the Kleenean inputs.

Finally, the function \mintinline{Haskell}{limit :: (Integer -> CReal) -> CReal} takes a fast converging sequence of real numbers and returns its limit.
We have to specify the type of the index \mintinline{Haskell}{n} since the limit function is in fact generic and can be applied to different types of sequences.

\subsection{Code extraction}

\begin{figure}[ht]
    \small\centering
    \begin{tabular}{c@{\kern 6em}c}
        \toprule
         Coq & Haskell \\
         \cmidrule(r){1-1}
         \cmidrule(l){2-2}
         \code{Real} & \code{AERN2.CReal} \\
         \code{Real0} & \code{0} \\
         \code{Realplus} & \code{(Prelude.+)}\\
         \code{limit} & \code{AERN2.limit} \\
         \code{choose} & \code{AERN2.select} \\
         \code{Realltb} & \code{(Numeric.OrdGenericBool.<)} \\
         \code{K} & \code{AERN2.CKleenean} \\
         \code{sumbool} & \code{Prelude.Bool} \\
         \code{M} & type identity \\
         \code{unitM} & \code{Prelude.id} \\
         \code{Nat.log2} & \code{(integer . integerLog2)} \\
         \bottomrule
    \end{tabular}
    \caption{Examples of our code extraction mappings}
    \label{fig:extraction-examples}
\end{figure}

Haskell/AERN code extraction is defined in file \texttt{Extract.v}. 
The key mappings are listed in Figure~\ref{fig:extraction-examples}.
Note that the monad $\mval$ does not appear in the extracted programs.  
Multivaluedness is intrinsic thanks to redundancy in the underlying representations.
The generalised comparison \code{Numeric.OrdGenericBool.<} returns the (lazy) Kleenean for real numbers.  It is the same comparison that features in Figure~\ref{fig:aern-realmax}.

Figure~\ref{fig:realmax-code} shows parts of the Coq proof of the maximisation operator specification, and relevant parts of the extracted Haskell code.

\begin{figure}[t]
    \centering
\begin{multicols}{2}
\begin{minted}
[fontsize=\fontsize{7}{9pt}\selectfont]
{Coq}
Lemma real_max_prop : 
   forall x y, {z | real_is_max x y z}.
Proof.
  intros.
  apply real_mslimit_P_lt.
  + (* max is single valued predicate *) 
    ...
  + (* construct limit *)
	intros.
	apply (mjoin (x>y - prec n)
	             (y > x - prec n)).
	++ intros [c1|c2].
	   +++ (* when y-2^n < x *)
         exists x. ...
	   +++ (* when x-2^n < y *)
         exists y. ...
	++ apply M_split.
	   apply @prec_pos.
Defined.
\end{minted}
\columnbreak
\begin{minted}
[fontsize=\fontsize{7}{9pt}\selectfont]
{Haskell}
realmax_prop ... x y =
  real_mslimit_P_lt ... (\n ->
    mjoin ...
      (\h -> 
        case h of {
          Prelude.True -> x;
          Prelude.False -> y})
      (m_split ... x y (prec ... n)))

m_split ... x y _UU03b5_ =
  choose ...
    (real_lt_semidec ...
      (real_minus ... y _UU03b5_) x)
    (real_lt_semidec ...
      (real_minus ... x _UU03b5_) y)
\end{minted}
%
\end{multicols}
    \vspace*{-2ex}
    \caption{Outline of a Coq proof and corresponding extracted Haskell code}
    \label{fig:realmax-code}
\end{figure}

\subsection{Performance measurements}
Since our axiomatization of constructive reals is built on a datatype similar to that used by AERN, we expect the performance of the extracted programs to be similar to that of hand-written AERN code.
Our measurements, summarized in Figure~\ref{fig:benchmarks},
are consistent with our hypothesis%
\footnote{Benchmarks were run 10 times on a Lenovo T440p laptop with Intel i7-4710MQ CPU and 16GB RAM, OS Ubuntu 18.04, compiled using Haskell Stackage LTS 17.2.}.
iRRAM is known to be one of the most efficient implementations of exact real computation and thus we also included hand-written iRRAM versions for calibration. 
The last three rows are examples of root finding by trisection.
The iRRAM trisection code benefits from in-place update.

\begin{figure}
    {\centering
    \begin{tabular}{ccccc}
        \toprule
        \multicolumn{2}{c}{Benchmark} & \multicolumn{3}{c}{Average execution time (s)}\\
        \cmidrule(r){1-2}
        \cmidrule(l){3-5}
        \small Formula & \kern 0.5em \small Accuracy\kern 0.5em & \kern 0.5em \small Extracted\kern 0.5em & \kern 0.5em \small Hand-written\kern 0.5em & \kern 0.5em\small iRRAM\\
        \cmidrule(r){1-1}
        \cmidrule(lr){2-2}
        \cmidrule(lr){3-3}
        \cmidrule(lr){4-4}
        \cmidrule(l){5-5}
        $\max(0,\pi-\pi)$ & $10^6$ bits & 16.8 & 16.2 & 1.59 \\[1.5ex]
        $\sqrt{2}$ & $10^6$ bits& 0.72 & 0.72 & 0.62 \\[0ex]
        $\sqrt{\sqrt{2}}$ & $10^6$ bits & 1.51 & 1.54 & 1.15 \\[1.5ex]
        $x-0.5=0$ & $10^3$ bits& 3.57 & 2.3 & 0.03 \\[0.5ex]
        $x(2-x)-0.5=0$ & $10^3$ bits& 4.30 & 3.08 & 0.04 \\[0.5ex]
        $\sqrt{x+0.5}-1=0$ & $10^3$ bits& 19.4 & 17.8 & 0.29 \\
        \bottomrule
    \end{tabular}\par}
    \vspace{1ex}
    \caption{Benchmarks and measurements}
    \label{fig:benchmarks}
\end{figure}
Figure~\ref{fig:csqrt-linechart} shows the execution times of this extracted code on a sample of inputs and with various target precisions.  
We use logarithmic scales to make the differences easier to see. 
Slower performance at zero reflects the fact that the limit computation uses the whole sequence, unlike away from zero where only a finite portion of the sequence is evaluated and the faster converging Heron iteration takes over for higher precisions.
The closer the input is to zero, the later this switch from limit to Heron takes place when increasing precision.  
For very large inputs, there is a notable constant overhead associated with scaling the input to the range where Heron method converges.  
The performance of our complex square root appears to be comparable to the performance of our real square root, which is comparable to a hand-written Haskell/AERN implementation  \cite{wollic}.

\begin{figure}
    \centering
    \includegraphics[width=\hsize]{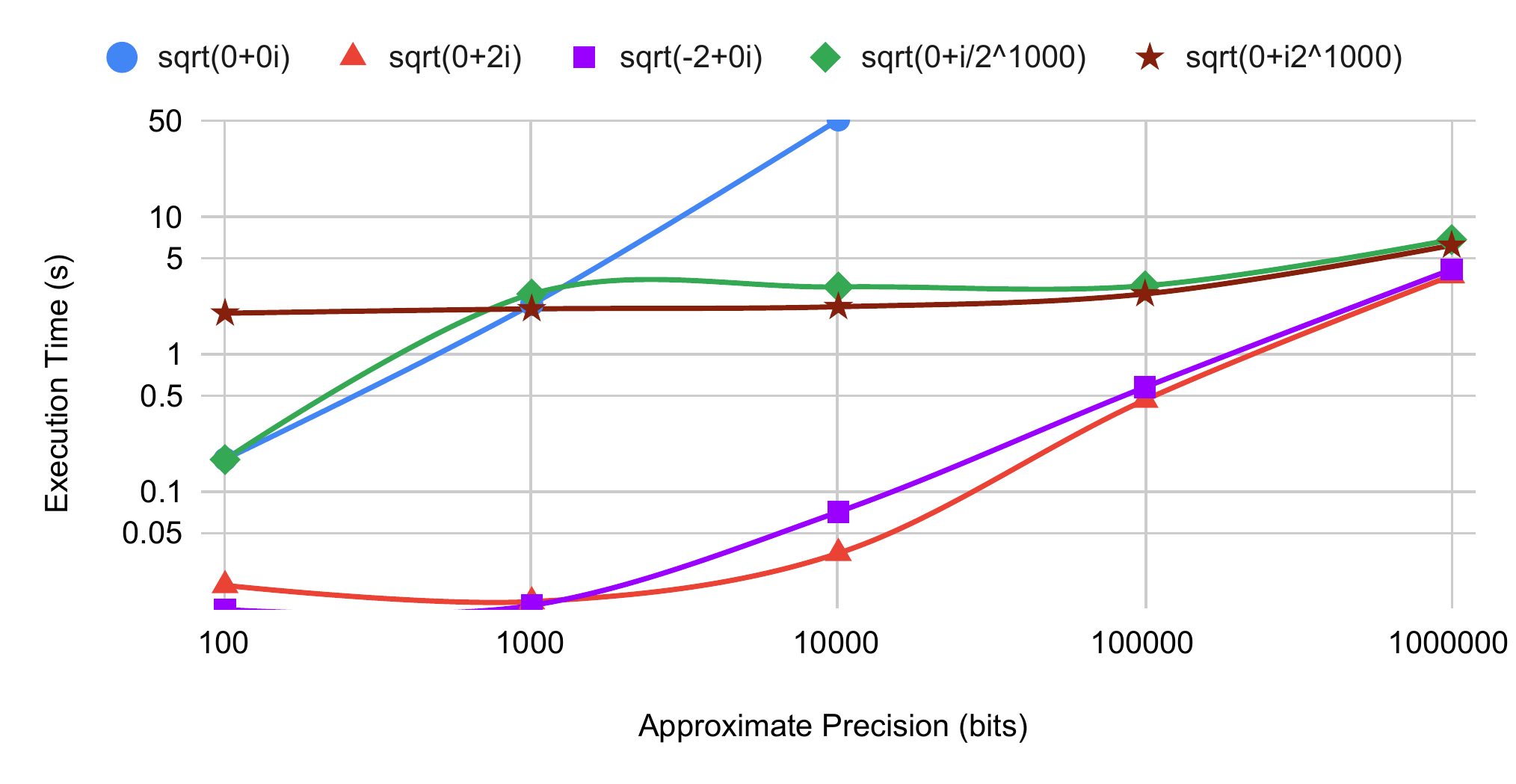}
    \caption{Execution time of the extracted complex square root function}
    \label{fig:csqrt-linechart}
\end{figure}
\section{Conclusion and Future Work}
We presented a new axiomatization of reals in a constructive dependent type theory and 
proved its soundness by extending the standard realizability interpretation from computable analysis.
We implemented our theory in Coq and used Coq's code extraction features to generate efficient Haskell programs 
for exact real computation based on the AERN library.

Our new axiomatization is particularly well-suited for verifying exact real computation programs 
where real numbers are provided as abstract entities and the underlying representations are hidden to the users.
Our axiomatization of reals is representation independent and its realizability interpretation is generically defined that
the interpretation can be setup to be any effective real numbers.
Hence, though in our example we mapped our axiomatic reals to a specific datatype of real numbers, \code{AERN2.CReal} in Haskell, 
it can be mapped also to any other sane datatype of reals.
Moreover, by proving that any concrete types satisfying our set of axioms are equivalent, 
it this sense, we showed that our axiomatization of reals is universal.


In future work, we plan to extend our implementation by other important operations on real numbers such as trigonometric and exponential functions and mathematical constants such as $\pi$ and $e$ and operations like polynomial root finding and matrix diagnonalization which essentially require nondeterministic limits.
Further, by extending our framework to real functions, 
including a sound axiomatization of continuity principle, e.g., \cite{xuthesis}, 
we could also study solution operators for ordinary or partial differential equations by applying recent ideas from real complexity theory \cite{kawamura2018parameterized,koswara2019computational}.

\bibliographystyle{alpha}
\bibliography{refs}

\newcommand{\etalchar}[1]{$^{#1}$}
\begin{thebibliography}{GWBC20}

\bibitem[BCC{\etalchar{+}}06]{balluchi2006ariadne}
Andrea Balluchi, Alberto Casagrande, Pieter Collins, Alberto Ferrari, Tiziano
  Villa, and Alberto~L Sangiovanni-Vincentelli.
\newblock Ariadne: a framework for reachability analysis of hybrid automata.
\newblock In {\em In: Proceedings of the International Syposium on Mathematical
  Theory of Networks and Systems.}, 2006.

\bibitem[Ber16]{berger2016extracting}
Ulrich Berger.
\newblock Extracting non-deterministic concurrent programs.
\newblock In {\em 25th EACSL Annual Conference on Computer Science Logic (CSL
  2016)}. Schloss Dagstuhl-Leibniz-Zentrum fuer Informatik, 2016.

\bibitem[BFC19]{boldo2019round}
Sylvie Boldo, Florian Faissole, and Alexandre Chapoutot.
\newblock {Round-off error and exceptional behavior analysis of explicit
  Runge-Kutta methods}.
\newblock {\em IEEE Transactions on Computers}, 69(12):1745--1756, 2019.

\bibitem[BFM09]{boldo2009}
Sylvie Boldo, Jean-Christophe Filli{\^a}tre, and Guillaume Melquiond.
\newblock Combining coq and gappa for certifying floating-point programs.
\newblock In Jacques Carette, Lucas Dixon, Claudio~Sacerdoti Coen, and
  Stephen~M. Watt, editors, {\em Intelligent Computer Mathematics}, pages
  59--74, Berlin, Heidelberg, 2009. Springer Berlin Heidelberg.

\bibitem[BH98]{BRATTKA1998490}
Vasco Brattka and Peter Hertling.
\newblock Feasible real random access machines.
\newblock {\em Journal of Complexity}, 14(4):490--526, 1998.

\bibitem[Bis67]{bishop1967foundations}
Errett~Albert Bishop.
\newblock Foundations of constructive analysis, 1967.

\bibitem[BLM16]{BLM16}
Sylvie Boldo, Catherine Lelay, and Guillaume Melquiond.
\newblock Formalization of real analysis: {A} survey of proof assistants and
  libraries.
\newblock {\em Mathematical Structures in Computer Science}, 26(7):1196--1233,
  2016.

\bibitem[BM11]{boldo2011flocq}
Sylvie Boldo and Guillaume Melquiond.
\newblock Flocq: A unified library for proving floating-point algorithms in
  coq.
\newblock In {\em 2011 IEEE 20th Symposium on Computer Arithmetic}, pages
  243--252. IEEE, 2011.

\bibitem[BM17]{DBLP:books/daglib/0041425}
Sylvie Boldo and Guillaume Melquiond.
\newblock {\em Computer Arithmetic and Formal Proofs - Verifying Floating-point
  Algorithms with the Coq System}.
\newblock {ISTE} Press, 2017.

\bibitem[Bra03]{Bra03f}
Vasco Brattka.
\newblock The emperor's new recursiveness: The epigraph of the exponential
  function in two models of computability.
\newblock In Masami Ito and Teruo Imaoka, editors, {\em Words, Languages \&
  Combinatorics III}, pages 63--72, Singapore, 2003. World Scientific
  Publishing.
\newblock ICWLC 2000, Kyoto, Japan, March 14--18, 2000.

\bibitem[Bri99]{BRIDGES199995}
Douglas~S. Bridges.
\newblock Constructive mathematics: a foundation for computable analysis.
\newblock {\em Theoretical Computer Science}, 219(1):95--109, 1999.

\bibitem[BT21]{DBLP:journals/apal/BergerT21}
Ulrich Berger and Hideki Tsuiki.
\newblock Intuitionistic fixed point logic.
\newblock {\em Ann. Pure Appl. Log.}, 172(3):102903, 2021.

\bibitem[Chl13]{chlipala2013certified}
Adam Chlipala.
\newblock {\em {Certified programming with dependent types: a pragmatic
  introduction to the Coq proof assistant}}.
\newblock MIT Press, 2013.

\bibitem[Coo08]{cooke2008classical}
Roger~L Cooke.
\newblock {\em Classical algebra: its nature, origins, and uses}.
\newblock John Wiley \& Sons, 2008.

\bibitem[FB18]{Mueller18}
Robert~Rettinger Franz~Brausse, Norbert~M{\"{u}}ller.
\newblock {Intensionality and Multi-Valued Limits}.
\newblock In {\em Proc. 15th Internat. Conf. on Computability and Complexity in
  Analysis ({CCA})}, page~11, 2018.

\bibitem[GWBC20]{gallois2020optimal}
Diane Gallois-Wong, Sylvie Boldo, and Pascal Cuoq.
\newblock Optimal inverse projection of floating-point addition.
\newblock {\em Numerical Algorithms}, 83(3):957--986, 2020.

\bibitem[Her99]{DBLP:journals/mlq/Hertling99}
Peter Hertling.
\newblock A real number structure that is effectively categorical.
\newblock {\em Math. Log. Q.}, 45:147--182, 1999.

\bibitem[Hof95]{10.1007/BFb0022273}
Martin Hofmann.
\newblock On the interpretation of type theory in locally cartesian closed
  categories.
\newblock In Leszek Pacholski and Jerzy Tiuryn, editors, {\em Computer Science
  Logic}, pages 427--441, Berlin, Heidelberg, 1995. Springer Berlin Heidelberg.

\bibitem[Jac99]{jacobs1999categorical}
Bart Jacobs.
\newblock {\em Categorical logic and type theory}.
\newblock Elsevier, 1999.

\bibitem[Kon18]{Konecny18}
Michal Konečný.
\newblock {Verified Exact Real Limit Computation}.
\newblock In {\em Proc. 15th Internat. Conf. on Computability and Complexity in
  Analysis ({CCA})}, pages 9--10, 2018.

\bibitem[Kon21]{konecny2008aern}
Michal Konečný.
\newblock {aern2-real: A Haskell library for exact real number computation}.
\newblock \url{https://hackage.haskell.org/package/aern2-real}, 2021.

\bibitem[KPT21]{wollic}
Michal Kone{\v{c}}ný, Sewon Park, and Holger Thies.
\newblock Axiomatic reals and certified efficient exact real computation.
\newblock In {\em International Workshop on Logic, Language, Information, and
  Computation}, pages 252--268. Springer, 2021.

\bibitem[KST18]{kawamura2018parameterized}
Akitoshi Kawamura, Florian Steinberg, and Holger Thies.
\newblock Parameterized complexity for uniform operators on multidimensional
  analytic functions and {ODE} solving.
\newblock In {\em International Workshop on Logic, Language, Information, and
  Computation}, pages 223--236. Springer, 2018.

\bibitem[KST20]{konevcny2020computable}
Michal Konečný, Florian Steinberg, and Holger Thies.
\newblock Computable analysis for verified exact real computation.
\newblock In {\em 40th IARCS Annual Conference on Foundations of Software
  Technology and Theoretical Computer Science (FSTTCS 2020)}. Schloss
  Dagstuhl-Leibniz-Zentrum f{\"u}r Informatik, 2020.

\bibitem[KSZ19]{koswara2019computational}
Ivan Koswara, Svetlana Selivanova, and Martin Ziegler.
\newblock Computational complexity of real powering and improved solving linear
  differential equations.
\newblock In {\em International Computer Science Symposium in Russia}, pages
  215--227. Springer, 2019.

\bibitem[KW85]{kreitz1985theory}
Christoph Kreitz and Klaus Weihrauch.
\newblock Theory of representations.
\newblock {\em Theoretical computer science}, 38:35--53, 1985.

\bibitem[Lon95]{longley1995realizability}
John~R Longley.
\newblock {\em Realizability toposes and language semantics}.
\newblock PhD thesis, University of Edinburgh. College of Science and
  Engineering., 1995.

\bibitem[Luc77]{LUCKHARDT1977321}
Horst Luckhardt.
\newblock A fundamental effect in computations on real numbers.
\newblock {\em Theoretical Computer Science}, 5(3):321 -- 324, 1977.

\bibitem[Mel08]{melquiond2008proving}
Guillaume Melquiond.
\newblock Proving bounds on real-valued functions with computations.
\newblock In {\em International Joint Conference on Automated Reasoning}, pages
  2--17. Springer, 2008.

\bibitem[Mon08]{monniaux2008pitfalls}
David Monniaux.
\newblock The pitfalls of verifying floating-point computations.
\newblock {\em ACM Transactions on Programming Languages and Systems (TOPLAS)},
  30(3):1--41, 2008.

\bibitem[MS15]{miyamoto2015program}
Kenji Miyamoto and Helmut Schwichtenberg.
\newblock Program extraction in exact real arithmetic.
\newblock {\em Mathematical Structures in Computer Science}, 25(8):1692--1704,
  2015.

\bibitem[M{\"u}l98]{muller1998implementing}
Norbert~Th M{\"u}ller.
\newblock Implementing limits in an interactive realram.
\newblock In {\em 3rd Conference on Real Numbers and Computers, 1998, Paris},
  volume~13, page~26, 1998.

\bibitem[M{\"u}l00]{irram}
Norbert~Th M{\"u}ller.
\newblock {The iRRAM: Exact arithmetic in C++}.
\newblock In {\em International Workshop on Computability and Complexity in
  Analysis}, pages 222--252. Springer, 2000.

\bibitem[NP18]{neumann2018topological}
Eike Neumann and Arno Pauly.
\newblock A topological view on algebraic computation models.
\newblock {\em Journal of Complexity}, 44:1--22, 2018.

\bibitem[PBC{\etalchar{+}}16]{park2016foundation}
Sewon Park, Franz Brau{\ss}e, Pieter Collins, SunYoung Kim, Michal Konečný,
  Gyesik Lee, Norbert M{\"u}ller, Eike Neumann, Norbert Preining, and Martin
  Ziegler.
\newblock Foundation of computer (algebra) analysis systems: Semantics, logic,
  programming, verification.
\newblock {\em arXiv e-prints}, pages arXiv--1608, 2016.

\bibitem[Reu99]{reus1999realizability}
Bernhard Reus.
\newblock Realizability models for type theories.
\newblock {\em Electronic Notes in Theoretical Computer Science},
  23(1):128--158, 1999.

\bibitem[RPS{\etalchar{+}}20]{ringer2020qed}
Talia Ringer, Karl Palmskog, Ilya Sergey, Milos Gligoric, and Zachary Tatlock.
\newblock {QED at large: A survey of engineering of formally verified
  software}.
\newblock {\em arXiv preprint arXiv:2003.06458}, 2020.

\bibitem[Sch02]{schroder2002effectivity}
Matthias Schr{\"o}der.
\newblock Effectivity in spaces with admissible multirepresentations.
\newblock {\em Mathematical Logic Quarterly: Mathematical Logic Quarterly},
  48(S1):78--90, 2002.

\bibitem[Sch06]{schwichtenberg2006constructive}
Helmut Schwichtenberg.
\newblock Constructive analysis with witnesses.
\newblock {\em Proof Technology and Computation. Natio Science Series}, pages
  323--354, 2006.

\bibitem[See84]{seely_1984}
R.~A.~G. Seely.
\newblock Locally cartesian closed categories and type theory.
\newblock {\em Mathematical Proceedings of the Cambridge Philosophical
  Society}, 95(1):33–48, 1984.

\bibitem[Spe49]{Spe49}
Ernst Specker.
\newblock Nicht konstruktiv beweisbare {S}\"atze der {A}nalysis.
\newblock {\em The Journal of Symbolic Logic}, 14(3):145--158, 1949.

\bibitem[Str91]{streicher2012semantics}
Thomas Streicher.
\newblock {\em Semantics of type theory: correctness, completeness and
  independence results}.
\newblock Progress in Theoretical Computer Science. Birkhäuser, Boston, MA,
  1991.

\bibitem[STT21]{steinberg2019quantitative}
Florian Steinberg, Laurent Thery, and Holger Thies.
\newblock {Computable analysis and notions of continuity in Coq}.
\newblock {\em {Logical Methods in Computer Science}}, {Volume 17, Issue 2},
  May 2021.

\bibitem[VO08]{van2008realizability}
Jaap Van~Oosten.
\newblock {\em Realizability: an introduction to its categorical side}.
\newblock Elsevier, 2008.

\bibitem[Wei00]{w00}
K.~Weihrauch.
\newblock Computable analysis.
\newblock Springer, Berlin, 2000.

\bibitem[Xu15]{xuthesis}
Chuangjie Xu.
\newblock {\em A continuous computational interpretation of type theories}.
\newblock PhD thesis, University of Birmingham, 2015.

\end{thebibliography}

\end{document}